\documentclass[oribibl]{llncs}
 
\usepackage{amssymb} % AMS Fonts
\usepackage{amsmath} % AMSLaTeX
\usepackage{url}
\usepackage{comment}
\usepackage{algorithm}
\usepackage{algorithmic}
\usepackage{graphicx}
\usepackage{epstopdf}
\usepackage{epsfig}
\usepackage{pdfsync}
\usepackage{multirow}

\makeatletter
\@twosidefalse
\makeatother

\newcommand{\KK}{{\mathbb{K}}}
\newcommand{\NN}{{\mathbb{N}}}
\newcommand{\ZZ}{{\mathbb{Z}}}
\newcommand{\FF}{{\mathbb{F}}}

\begin{document}

\title{A New Primitive for a Diffie-Hellman-like Key Exchange Protocol Based on
  Multivariate Ore Polynomials}

\author{Reinhold Burger and Albert Heinle}

\institute{Symbolic
  Computation Group\\ David R. Cheriton School of Computer Science\\
University of Waterloo, Waterloo, Canada\\
Email: \{rfburger, aheinle\}@uwaterloo.ca
}

\maketitle

\begin{abstract}

In this paper we present a new primitive for a key exchange protocol based on multivariate
non-commutative polynomial rings, analogous to the classic
Diffie-Hellman method. Our technique extends the proposed scheme of
Boucher et al. from 2010. Their method was broken by Dubois and
Kammerer in 2011, who exploited the Euclidean domain structure of the
chosen ring.  However, our proposal is immune against such attacks,
without losing the advantages of non-commutative polynomial rings as
outlined by Boucher et al.  Moreover, our extension is not restricted
to any particular ring, but is designed to allow users to readily
choose from a large class of rings when applying the protocol. Our
primitive can also be applied to other cryptographic paradigms. In
particular, we develop a three-pass
protocol, a public key cryptosystem, a digital signature scheme and a
zero-knowledge proof protocol.

 % Moreover,
% we present a general definition of non-commutative polynomial rings
% which can be the .
% However, our method is secure against the attack noticed by
% Dubois and Kammerer in 2011.

\end{abstract}

\section{Introduction}
\label{sctn:intro}

In 2010, Boucher et al. \cite{boucher2010key} proposed a novel
Diffie-Hellman-like key exchange protocol \cite{diffie1976new} based on skew-polynomial rings. An outline of
their method can be given as follows (with adapted notation): Two
communicating parties, Alice and Bob, publicly agree on an element $L$ in
a predetermined skew-polynomial ring, and on a subset $\mathcal{S}$ of commuting
elements in this ring. Alice then chooses two private keys $P_A,Q_A$ from
$\mathcal{S}$ and sends Bob the product $P_A\cdot L \cdot Q_A$. Bob
similarly chooses $P_B,Q_B$ from $\mathcal{S}$ and sends Alice
$P_B\cdot L \cdot Q_B$. Alice computes $P_A\cdot P_B\cdot L\cdot
Q_B\cdot Q_A$, while Bob
computes $P_B\cdot P_A \cdot L \cdot Q_A\cdot Q_B$. Since $P_A\cdot P_B =
P_B\cdot P_A$, and $Q_A\cdot Q_B =Q_B\cdot Q_A$, Alice and Bob have
computed the same final element, which can be used as a secret key,
either directly or by hashing. Boucher et al. claimed that it would be
intractable for an eavesdropper, Eve, to compute this secret key with
knowledge only of $L$, $\mathcal{S}$, $P_A\cdot L \cdot Q_A$ and
$P_B\cdot L \cdot Q_B$. They based their claim on the difficulty of
the factorization problem in skew-polynomial rings, in particular the
non-uniqueness of factorizations.

However, in 2011, Dubois and Kammerer exploited the fact that the
concrete skew-polynomial ring chosen by Boucher et al. is a Euclidean
domain to successfully attack their protocol \cite{dubois2011cryptanalysis}. Following their
approach, an eavesdropper Eve chooses a random element $e \in \mathcal{S}$, and
computes the greatest common right divisor of
$P_A\cdot L\cdot Q_A\cdot e = P_A\cdot L\cdot e\cdot Q_A$
with $P_A\cdot L\cdot Q_A$, which is with high probability equal
to $Q_A$. From
this point on, Eve can easily recover the agreed upon key between
Alice and Bob. Moreover, the authors also criticized the suggested
brute-force method for Alice and Bob to generate commuting
polynomials, as most of the commuting polynomials turn out to be central
and thus the possible choices for private keys becomes fairly small.
(An element of a ring is \textbf{central} if it commutes with all
other elements of the ring.)

After Dubois and Kammerer's paper, interest in the application of
non-commutative polynomial rings appears to have dwindled, since to
the best
of our knowledge no
further publications considering non-commutative polynomial rings in
cryptographic contexts have appeared. % It seems that researchers
% have concluded that the use of non-commutative rings 

It is our position that such rings can still be used
as a foundation of a secure Diffie-Hellman-like protocol. The basic weakness in
the scheme presented in \cite{boucher2010key} lies in the choice of a
univariate Ore extension as the underlying ring of the protocol, as
these rings are Euclidean domains. However, the construction of
Ore extensions can be iterated, and the resulting multivariate Ore
polynomial rings will no longer be principal ideal domains (and
therefore not Euclidean domains). This would preclude any
attack of the type proposed by Dubois and Kammerer.

\begin{comment}
Other Pubs:
 - Tsarev, van Hoeij, REDUCE, Beals Kartashova, Giesbrecht, Singular,
   Grigoriev-Schwarz, Bueso-Gomez-Verschoren, Ore
 - Ulmer & Boucher & Co.
 - Burchnall Chaundry
 - Diffie Hellman
Motivation:
 - Paper by Ulmer and Boucher did not work due to simple attack
 - Attack not possible if no principal ideal domain.
 - Critique/Improvements for finding keys (feasibly).
   - Finding commutative polynomials
Background:
 - Construction and motivation of Ore-Polys
 - Factorization problem
   - UFD up to similarity (similarity is involved)
   - Principal ideal/Euclidean domain (when?)
   - infinitely many factorizations possible
Definitions/Notations:
 - Sigma-derivations
 - Ore-extensions
 - $n^{\mathrm{th}}$ Weyl algebra/Shift/$\underline{q}$-Weyl/quantum
\end{comment}

The contributions of this paper are the following:

\begin{itemize}
  \item We present a method of constructing non-commutative algebras
    to use in the protocol as
    presented in \cite{boucher2010key}. The creation of algebras in
    this fashion ensures that the
    Diffie-Hellman-like key exchange will not be subject to attacks as described
    in \cite{dubois2011cryptanalysis}. At the same time the desirable 
    properties such as non-uniqueness of the factorization remain
    present, as well as key-generation in a feasible amount of
    time.
  \item For our choice of non-commutative algebras, no polynomial
    time factorization algorithm for their elements is known. For most
    of them,
    there is not even a general factorization technique, i.e. one
    that is applicable to any
    element,
    discovered yet and, as some of them do not even have the
    property of being Noetherian, factorization algorithms may not
    even exist. % This undecidability is, we would conjecture,
    % true for both finite and
    % zero-characteristic underlying fields.
  \item This paper addresses the critique given in
    \cite{dubois2011cryptanalysis} concerning the feasible
    construction of a set $\mathcal{S}$ of commuting, non-central
    elements, from which the two communicating parties choose their private
    keys. We show an efficient way to construct commuting
    polynomials, which is independent from the choice of the algebra.
  \item We have made an experimental implementation of the
    key-exchange protocol using our proposed rings, in order to
    examine the practical feasibility of our primitive. Furthermore,
    we have created challenge problems for the reader who wishes to
    examine the security of our schemes.
  \item Attacks based on the key-choice of $A$ and $B$ are studied for
    a concrete algebra
    and an overview of weak keys and their detection is presented.
  \item We also study the application of multivariate Ore polynomials to other cryptographic
    paradigms: a three-pass protocol, a public key cryptosystem, a
    digital signature scheme and a zero-knowledge proof protocol.
\end{itemize}

In the rest of this section, we will introduce some basic notations
and definitions. Moreover, we present the main ring structures which will be used
for the various cryptographic protocols. Furthermore,
related work and the potential for these rings to be used in
post-quantum cryptosystems will be discussed.

In section \ref{sctn:dhkep}, we show how multivariate Ore polynomial rings can be
applied as a primitive for a Diffie-Hellman-like key
exchange protocol. We will
argue its correctness, its efficiency and its security.

Section \ref{sctn:impl} describes an implementation
of our proposal for a specific choice of the non-commutative
ring. Experimental results are presented for different input-sizes.

In section \ref{sctn:insecurekeys}, we discuss some known insecure
keys for a particular ring, namely the $n^{\mathrm{th}}$ Weyl
algebra. The described construction cannot always be generalized to
other rings,
albeit it provides a guideline for the study of insecure keys in rings that can be used
in cryptographic protocols.

Before our conclusion in section \ref{sctn:conclusion}, we present in
section \ref{sec:Enhancements} applications of multivariate
Ore polynomial rings to
other cryptographic paradigms.

\subsection{Basic Notations and Definitions}
\label{subsec:basic-notations-and-definitions}

Throughout the whole paper, $R$ denotes an arbitrary domain with
identity. For practical reasons we furthermore assume for any introduced ring
that it is
computable, i.e. one can find a finite representation of its
elements, and that all arithmetics can be done in polynomial
time. Also a random choice of an element in the ring $R$ is assumed
to be possible with polynomial costs.

Let us address the basic construction principles of so-called Ore
extensions \cite{ore1933theory}
of $R$. We follow the notions from \cite{Bueso:2003}, which
we also recommend as a resource for a thorough introduction into
the field of algorithmic non-commutative algebra.

\begin{definition}[\cite{Bueso:2003}, Definition 3.1]
\label{def:sigmaDerivation}
Let $\sigma$ be a ring endomorphism of $R$. A \textbf{$\sigma$-derivation} of
$R$ is an additive endomorphism $\delta: R \to R$ with the property
that $\delta(rs) = \sigma(r)\delta(s) + \delta(r)s$ for all $r,s \in
R$. We call the pair $(\sigma, \delta)$ a \textbf{quasi-derivation}.
\end{definition}

\begin{proposition}[\cite{Bueso:2003}, Proposition 3.3.]
  \label{prop:ringExtensionQuasiDerivation}
  Let $(\sigma, \delta)$ be a quasi-derivation on $R$. Then there
  exists a ring $S$ with the following properties:
  \begin{enumerate}
    \item $R$ is a subring of $S$;
    \item there exists an element $\partial \in S$ such that $S$
      is freely
      generated as a left $R$-module by the positive powers
      $1, \partial, \partial^2, \ldots$ of $\partial$;
    \item for every $r \in R$, we have
      $\partial r = \sigma(r)\partial + \delta(r)$.
  \end{enumerate}
\end{proposition}

\begin{definition}[cf. \cite{Bueso:2003}, Definition 3.4]
  \label{def:oreExtension}
  The ring $S$ defined by the previous result and denoted by
    $R[\partial\,;
    \sigma, \delta]$ is usually referred to as an \textbf{Ore
    extension of $R$}.
\end{definition}

\begin{comment}
\begin{remark}
The requirement that $\sigma(r)$ be an endomorphism, and
$\delta(r)$ be an additive endomorphism, ensures that
various simplifications of an Ore polynomial expression will yield
the same final result.  For example, if $s_{1},s_{2},s_{3}\in S$,
then the expression $r_{1}\cdot \left(r_{2}+r_{3}\right)$ could
be evaluated by first adding $r_{2}$ and $r_{3}$, and then
multiplying (on the left) by $r_{1}$.  Alternatively, one could
expand this as $r_{1}r_{2}+r_{1}r_{3}$, and evaluate by finding
the two products and finally adding them.  The endomorphism
conditions on $\sigma$ and $\delta$ ensure that the final result
of these two methods will be equal.
\end{remark}
\end{comment}

\textbf{General Assumption:} As we want the Ore extension to have at least the property of being a
domain, we assume from now on that $\sigma$ is always injective (compare
\cite{Bueso:2003}, Proposition 3.10). In order to keep the costs of
arithmetics in $R[\partial\,;\sigma,\delta]$ polynomial, we make
two additional assumptions:
\begin{itemize}
\item[(i)] There exist polynomial time algorithms to compute
$\sigma(r)$ and $\delta(r)$ for any given $r \in R$.
\item[(ii)] Either $\sigma$ is the identity map, or $\delta$ is the
  zero map.
\end{itemize}
While item (i) seems to be a natural assumption, item (ii) may seem highly
restrictive. But these cases cover several algebras that are studied
in practice, as we will point out in the examples below. The need for
this condition comes from the result of the following lemma, which can
be easily proven by induction on $n \in \NN$.

\begin{lemma}
\label{lem:multCoeffWithMonom}
  Let $R[\partial\,;\sigma, \delta]$ be an Ore extension of $R$,
  and let $r$ be
  an arbitrary element in $R$. Then we have the following identity for
  $n \in \NN$:
  \begin{eqnarray*}
  \partial^{n}r &=& \sigma^n(r)\partial^{n}
             + \Biggl( \sum_{\theta \in S_n \bullet
             [\underbrace{\sigma,\ldots,\sigma}_{n-1 \text{\ }times},
             \delta]}\theta_1 \circ \ldots \circ \theta_n \circ r \Biggr)
             \partial^{n-1}+\ldots\\
                 & & +\Biggl( \sum_{\theta \in S_n \bullet
                   [\sigma,\delta,\ldots,\delta]}\theta_1 \circ \ldots
                 \circ \theta_n \circ r \Biggr)\partial + \delta^n(r),
  \end{eqnarray*}
where $S_n$ denotes the permutation group on $n$ elements and $\bullet$
the canonical action of the group on a list with $n$ elements.
\end{lemma}

Without item (ii), when
multiplying elements in $S$, we would have to compute up to $2^n$
images of an
element $r \in R$ resulting from all different ways of applying $n$
times functions chosen from the set $\{\sigma, \delta\}$. This is avoided by choosing one of
the maps to be trivial, i.e. by our assumption (ii).

\begin{example}
For the Ore extensions considered in the paper \cite{boucher2010key},
the authors assumed $R$ to be a finite
field, $\sigma$ to be a power of the Frobenius automorphism on $R$ and $\delta$ to be the zero map.
\end{example}

\begin{example}
  The construction of a commutative polynomial ring over a given ring
  $R$ can be viewed as an Ore extension by choosing $\sigma$ to be the
  identity map and $\delta$ to be the zero map.
\end{example}

\begin{remark}
  If we choose $\sigma$ not to be an automorphism, then our
  constructed ring is not necessarily Noetherian, which makes the general
  factorization problem even harder to solve. An example
  of a non-Noetherian Ore extension is the following:\\
Let $\KK$ be a field. Set $R :=
\KK[y]$, the univariate polynomial ring over $\KK$. Define
$\sigma: R \to R, f(y) \mapsto f(y^2)$ and set $\delta$ to be the zero
map. Then
$(\sigma, \delta)$ is a quasi-derivation, and the ring
$R[\partial\,; \sigma, \delta]$ is not Noetherian.
A proof of this, and a more thorough
discussion, can be found in \cite{mcconnell2001noncommutative},
section 1.3.2.
\end{remark}

The process of building an Ore extension can be
iterated. The rings that we propose to use for a
key exchange protocol 
are of the form
\begin{align}
\label{eq:ourRings}
S:= R[\partial_1;\sigma_1,\delta_1] [\partial_2;\sigma_2,\delta_2]
\ldots [\partial_n;\sigma_n,\delta_n],
\end{align}
where $\NN \ni n > 1$, $R$ is a domain with identity element, and for all $i \in
\{1, \ldots, n\}$, either $\sigma_i$ is the identity map, or $\delta_i
$ is the zero map, according to our general assumptions. We will refer
to these rings throughout the paper as \textbf{rings of type
  (\ref{eq:ourRings})}. Furthermore, the $\sigma_i$ are
injective, and there exists a subring $\tilde R \neq \{0\}$ of the
center of $R$, such
that for all $i \in \{1,\ldots,n\}$ and $r \in \tilde R: \sigma_i(r)
= r$ and $\delta_i(r) = 0$. We refer to $\tilde R$ as the \textbf{subring
of constants}.

The condition $n>1$ gives us the property that our ring is neither a
left nor a right
principal ideal domain and therefore there exists no notion of a left-
or right greatest common divisor. Thus, our construction of a
Diffie-Hellman-like key exchange protocol would not be vulnerable to
the methods introduced in \cite{dubois2011cryptanalysis}.%   (The proof
% of this property of S is rather simple.
% But as this property is
% important for the security of our system, a proof is provided
% in Appendix \ref{app:S-not-Euclidean}.)

The condition that $\tilde R$, whose elements are not subject to any
non-commutative relation,  exists, is needed later to construct commutative
subrings in rings of type (\ref{eq:ourRings}).

There will be two kinds of rings of type (\ref{eq:ourRings}) that will serve as model examples
throughout the paper:

\begin{definition}
  \label{def:weyl}
  The so-called \textbf{\boldmath$n^{\mathrm{th}}$ Weyl algebra $A_n$}
  is an Ore extension of type
  (\ref{eq:ourRings}). For this, let $\tilde \KK$ be an arbitrary field,
  and
  $\KK := \tilde \KK(x_1,\ldots,x_n)$. Define for all $i \in
  \{1,\ldots, n\}$ the
  $\sigma_i$ to be identity maps, and define $\delta_i$ to be the
  partial derivation with respect to $x_i$.
  Thus, $\partial_{i}x_{i} = \sigma_{i}(x_{i})\partial_{i}
   + \delta_{i}(x_{i}) = x_{i}\partial_{i} + x_{i}'
   = x_{i}\partial_{i} + 1$.  But $\partial_{i}$ commutes with all
   other $x_{j}$, where $j\ne i$.  Also, $\partial_{i}$, $\partial_{j}$
   always commute, as do $x_{i}$, $x_{j}$. Finally $\tilde \KK$ is the
   subring of constants of $S$.
\end{definition}

\begin{example}
  \label{ex:boucherExtended}
  The rings used in \cite{boucher2010key}, namely
  $\mathbb{F}_q[\partial\,;\sigma]$ with
  $\sigma \in \mathrm{Aut}(\mathbb{F}_q)$, are single Ore
  extensions of a finite field $\mathbb{F}_q$, where $q$ is some positive power
  of a prime number $p$. One can iterate the extension of this ring,
  and create
  $\mathbb{F}_q[\partial_1;\sigma_1] \ldots[\partial_n ; \sigma_n]$,
  where $\sigma_1, \ldots, \sigma_n \in \mathrm{Aut}(\mathbb{F}_q)$.
  (Thus, if $c\in \mathbb{F}_q$, then
  $\partial_{i}c = \sigma_{i}(c)\partial_{i}$.  But $\partial_{i}$,
  $\partial_{j}$ commute, for all $i,j$.)
\end{example}

The factorization properties of a ring $S$ that has the form
(\ref{eq:ourRings})  are
different from commutative multivariate polynomial rings. In
particular, $S$ is not a unique factorization domain in the classical
sense, i.e. factors are not just unique up to permutation and
multiplication by units. Factors in $S$ are unique up to the following
notion of similarity.

\begin{definition}[cf. \cite{Bueso:2003}, Definition 4.9]
\label{definition:similarity}
  Let $R$ be a domain. Two elements $r,s\in R$ are said to be
  \textbf{similar}, if $R/Rr$ and $R/Rs$ are isomorphic as left $R$-modules.
\end{definition}

In general, given an element $p \in S$, which has two different
complete factorizations
\[
p = p_1 \cdots p_m = \tilde p_1 \cdots \tilde p_{k},
\]
where $m, k \in \NN$, there exists
for every $i \in \{1, \ldots, m\}$ at least one $j \in \{1,\ldots, k\}$
such that $p_i$ is similar to $\tilde p_j$
(cf. \cite{nathan1943theory}, \cite{ore1933theory}). This means, that
the position of similar elements in different factorizations is not fixed.

\begin{example}
  Let us state  a simple example in $A_2$, that can be found in
  \cite{landau1902satz}:
  \[
  h := (\partial_1+1)^2(\partial_1 + x_1\partial_2) \in A_2.
  \]
  Besides the given factorization in the definition of $h$, we have
  the following decomposition into irreducible elements:
  $$h =
  (x_1\partial_1\partial_2+\partial_1^2+x_1\partial_2+\partial_1+2\partial_2)(\partial_1+1).$$
\end{example}

The corresponding decision problem, of deciding whether two given
polynomials
are similar, is not known to be possible in polynomial time to the best of our
knowledge, although attempts have been made \cite{cha2010solving}.

One possible attack is to factor elements in non-commutative
polynomial rings. However, besides the fact that factoring is
currently intractable in this setup, the non-uniqueness of the factorization
adds another difficulty for an attack based on factorization. As the next example illustrates, one might end up with
infinitely many factorizations, from which one has to choose the
correct one.

\begin{example}
Let $\KK$ in Definition \ref{def:weyl} be of characteristic
zero and consider
$\partial_1^2 \in A_n$. Besides factoring into
$\partial_1\cdot \partial_1$, it also factors for all $c \in \KK$ into
$$\left( \partial_1 + \frac{1}{x_1+c}\right)
\cdot \left( \partial_1 - \frac{1}{x_1+c} \right).$$
\end{example}

\begin{remark}
  In Algorithm \ref{alg:keyExchange}, we are exclusively using the multiplicative structure of
  the ring $S$. Only the generation of random elements causes us to
  apply addition. It is to be emphasized, that $(S,\cdot)$, i.e. the
  elements in $S$ equipped with the multiplication operator, do not
  form a group but a monoid. This is due to the fact that the
  variables introduced by the Ore extension (i.e. $\partial_1, \ldots,
  \partial_n$) do not have a multiplicative inverse.
\end{remark}

\subsection{Potential as a Post-Quantum Cryptosystem}
\label{subsec:potential-as-a-post-quantum-ryptosystem}
Here we will try to give some plausible grounds for our conjecture,
that the factorization problem for our rings cannot be solved
in polynomial time, even with quantum algorithms.

This stems from the observation that factors of an Ore polynomial
$p$ can be very large compared with $p$ itself.  Indeed, in terms
of bit length representations, the size of the factors can be
exponential in the size of $p$.  For example, consider the
Chebyshev differential operator
\begin{equation}  \label{eq:chebyshev_L}
  L := \left(1-x_{1}^{2}\right) \partial_{1}^{2}
       - x_{1} \partial_{1} + n^{2},
\end{equation}
 as an element in $A_1$, where $n$ is a real constant.  When $n$ is a positive integer,
one can show that $L$ has two possible factorizations.
Furthermore, when $n$ is prime one can show that these factors will
contain $\Omega (n)$
non-zero terms in expanded representation.  Thus their bit size grows at least as fast
as $n$,
while the size
of $L$ itself grows only as $\log n$.  Consequently, the sizes of the factors
are exponentially
larger than the size of $L$.  If the reader wishes to try some
experimentation in \textsc{Maple}, we provide a code snippet in
appendix \ref{app:Chebyshev}.

Now, for the decision problem, ``Is $L$ factorable?", the obvious
certificate for verification of a ``yes" answer would be an
actual pair of factors of $L$.  But as we can see from this example,
the size of such a
certificate may not be polynomial in the input size of $L$.
Furthermore, this problem is already occurring for the simplest
possible case: a second-order operator in a univariate Ore ring.
Our proposal works with much higher-order operators over many
variables, so the relative size of such a certificate of factors
will not improve, and may even become worse.
 
Of course, this does not prove that a polynomial-sized certificate
could not exist.  But we do not know of any, and hence we
suspect that this problem may not even belong to the class $NP$.
As there is some thought that NP-complete problems would not have
polynomial time quantum algorithms (see e.g. \cite{brown2001quest},
page 297), we are therefore led to conjecture that
our factorization problem would not have any such algorithm, either.
 
Note, though, that the above example was over a field of
characteristic zero.  We actually prefer to work over finite fields,
to reduce expression swell in the computations. For such fields, we do
not know of any examples where the bit-size of the factors is
exponentially larger than the input-polynomial. However, even for
univariate differential polynomials, there are no known polynomial
time algorithms for factorization or deciding irreducibility, for  either
classical or quantum computers.
For Noetherian rings over finite fields,
the hardness of factorization is
less clear, though there are still no known polynomial time
algorithms for the multivariate case.  But even if one exists
for Noetherian rings over finite fields, one
could instead choose a ring having a non-Noetherian extension.
As mentioned above, we are skeptical that there is any
polynomial time algorithm for this case, using a classical computer.
Furthermore, unless there is some property of the non-Noetherian ring
that a quantum algorithm
can take advantage of, we conjecture the same is true for
quantum computers.

% But again, even if efficient
% factorization algorithms were ever found for all finite field
% cases, our protocol could instead
% be used with non-Noetherian non-commutative polynomial
% rings over a field of characteristic zero.  The resulting
% key sizes would be very large, but we believe it would be quite
% secure against a direct factorization attack.  It would be
% very surprising indeed if there were a polynomial time factorization
% algorithm for such an extreme ring, even using a quantum computer.

% One particular approach may be mentioned here: based on degree
% bounds, one may set up an ansatz of the factors with unknown
% coefficients.  Substitution leads to a nonlinear system of
% multivariate polynomial equations in the unknowns.
% Algorithms for solving these systems are known, involving
% Gr\"{o}bner bases.  However, these have no guarantee of
% polynomial runtime, and indeed, examples are known where
% the size of a computed Gr\"{o}bner basis is doubly-exponential
% in the size of the input.  Furthermore, it can be shown that
% an arbitrary instance of 3-SAT can be transformed into a
% rather simple looking polynomial system, involving only linear
% and quadratic equations, which will therefore be NP-hard
% to solve.  We therefore are doubtful that any approaches
% which involve solving an ansatz will yield a polynomial time
% algorithm, even with a quantum computer.

\subsection{Related Work}

%Factorization Section

There exist a polynomial time algorithms that factor
univariate Ore-polynomials over finite fields, namely
\cite{giesbrecht1998factoring,giesbrecht2003factoring}. This includes the skew-polynomials as
used in \cite{boucher2010key}. Boucher
et al. argued that even if an attacker can find a factorization using
this algorithm, then it might not be the right one to discover the key
$A$ and $B$ agreed upon. This can be true for certain choices of
polynomials, but there is more theory needed to prove that there is a
certain lower bound on the number of different factorizations.

For certain single Ore extensions of a univariate commutative
polynomial ring or function
field there are several algorithms and even implementations
available. This is due to the fact that those extensions are algebraic
generalizations of operator algebras. The most prominent publications that
deal with factoring in the first Weyl algebra are
\cite{van1997formal,Hoeij:1996,Hoeij:1997,Hoeij:2010,GrigorievSchwartz:2004}; the algorithms of the first four papers are
implemented in the computer algebra system
\textsc{Maple} \cite{Maple}, and that of the fifth paper in
\textsc{ALLTYPES} \cite{Schwarz:2009}. For factoring elements in the
first Weyl algebra with
polynomial coefficients, there is an implementation
\cite{heinle2013factorization} in the
computer algebra system \textsc{Singular} \cite{Singular:2012}. The
implementation also extends to the shift algebra and classes of
polynomials in the so-called first $\underline{q}$-Weyl algebra.
Theoretical results for those operator algebras are shown in
\cite{Tsarev:1994} and \cite{Tsarev:1996},  which extend the papers
\cite{Loewy:1903} and \cite{Loewy:1906}.

The factorization problem for general multivariate Ore algebras has
not, as of yet, been as well investigated.

A thorough theoretic overview of the factorization problem in Ore
domains is presented in \cite{Bueso:2003}.

Recently, the
techniques from \cite{heinle2013factorization} were extended to factor elements
in the $n^{\mathrm{th}}$ Weyl
algebra, the $n^{\mathrm{th}}$ shift algebra and classes of polynomials in the
$n^{\mathrm{th}}$ $\underline{q}$-Weyl algebra \cite{giesbrecht2014factoring}.
However, the
algorithm uses Gr\"obner-bases \cite{Buchberger:1997}, and therefore does not run
in polynomial time \cite{mayr1982complexity}.

From an algebraic point of view, and dealing only with strictly
polynomial non-commutative algebras, Melenk and Apel \cite{Melenk:1994}
developed a package for the computer algebra system
\textsc{REDUCE}. That package provides tools to deal with certain
non-commutative polynomial algebras and also contains a factorization
algorithm for the supported algebras.

Beals and Kartashova
\cite{beals2005constructively} consider the problem of factoring
polynomials in the second Weyl algebra, where they are able to deduce
parametric factors. Research in a similar direction was done by
Shemyakova in
\cite{parameters:shemyakova:2007,shemyakova:multfacts1:2009,2010:shemyakova:refinement}.

%Other non-commuta Protocol

Another key exchange protocol based on non-commutative rings is
presented in \cite{climent2012key}. The ring chosen in this
publication is the ring of endomorphisms of $\ZZ_p \times \ZZ_{p^2}$,
which is also not a principal ideal domain and therefore not subject
to an attack as shown in \cite{dubois2011cryptanalysis}. It should be
noted that the authors used the same technique as in this paper
for constructing commuting elements.

Using a different non-commutative ring, which is a $\ZZ$-module, Cao et al. have
presented a similar key-exchange protocol in \cite{cao2007new}. The
authors also use the same idea to construct commuting elements, but
their work furthermore considers some other non-abelian groups. We are
not aware of any known attack on this system.

% Group protocols.
An approach for using non-abelian groups to generate a public key cryptosystem was
developed by Ko et al. in \cite{ko2000new} for the special case of
braid groups \cite{artin1947theory}. The authors used the conjugacy
problem of groups as their hard problem. However, Jun and Cheon
presented in \cite{cheon2003polynomial} a polynomial time algorithm
for exactly their setup (but not for the conjugacy problem in
general). This attack exploits the Lawrence-Krammer representation of
braid groups \cite{krammer2002braid}, which is a linear representation
of the braid group.

%Commutative Polynomials

Concerning the task of finding commuting polynomials in the first Weyl
algebra, a very thorough study is presented in
\cite{burchnall1923commutative}, which also demonstrates the hardness
to find all commuting polynomials in the ring of ordinary linear differential operators.

\section{The Key Exchange Protocol}
\label{sctn:dhkep}

\subsection{Description of the Protocol}
\label{subsec:description-of-the-protocol}

We refer to our communicating parties as Alice (abbreviated $A$) and Bob
(abbreviated $B$). Alice and Bob wish to agree on a common secret
key using a Diffie-Hellmann-like cryptosystem.

The main idea is similar to the key exchange protocol presented in
\cite{boucher2010key}. The main differences are that (i) the ring $S$ and (ii) the commuting
subsets are not fixed, but agreed upon as part of the key-exchange protocol. It is summarized by the
following algorithm.

\begin{algorithm}[H]
  \begin{algorithmic}[1]
    \STATE $A$ and $B$ publicly agree on a ring $S$ of type
    (\ref{eq:ourRings}), a security parameter $\nu \in \NN$ representing the
    size of the elements to be picked from $S$ in terms of total
    degree and coefficients, a non-central element $L \in S$, and two multiplicatively
    closed, commutative subsets of $\mathcal{C}_l,\mathcal{C}_r
    \subset S$, whose elements do not commute
    with $L$. \label{stp:public}
    \STATE $A$ chooses a tuple $(P_A,Q_A) \in \mathcal{C}_{l}
    \times \mathcal{C}_{r}$.\label{stp:secretA}
    \STATE $B$ chooses a tuple $(P_B,Q_B) \in \mathcal{C}_{l}
    \times \mathcal{C}_{r}$. \label{stp:secretB}
    \STATE $A$ sends the product $A_{\text{part}} := P_A \cdot L \cdot
    Q_A$ to $B$.
    \STATE $B$ sends the product $B_{\text{part}} := P_B \cdot
    L \cdot Q_B$ to $A$.
    \STATE $A$ computes $P_A \cdot B_{\text{part}} \cdot Q_A$.
    \STATE $B$ computes $P_B \cdot A_{\text{part}} \cdot Q_B$.
    \STATE $P_A \cdot P_B \cdot
    L \cdot Q_B \cdot Q_A = P_B \cdot
    P_A \cdot L \cdot Q_A \cdot Q_B$ is the shared secret key of $A$
    and $B$.\label{stp:finalAlg1}
  \end{algorithmic}
\caption{DH-like protocol with rings of type (\ref{eq:ourRings})}
  \label{alg:keyExchange}
\end{algorithm}

\begin{proof}[Correctness of Algorithm \ref{alg:keyExchange}]
  As $P_A, P_B \in \mathcal{C}_l$ and $Q_A, Q_B\in \mathcal{C}_r$, we
  have the identity in step \ref{stp:finalAlg1}. Therefore, by the end
  of the key exchange, both $A$ and $B$ are in possession of the same
  secret key.
\end{proof}

Before discussing the complexity and security of the proposed scheme, we
deal with the feasibility of constructing the sets $\mathcal{C}_{l}, \mathcal{C}_{r}$ in
Algorithm \ref{alg:keyExchange}. We propose the following technique,
which is applicable independent of the choice of $S$. Let $P, Q \in S$
be chosen, such that they do not commute 
with $L$. Define
\begin{eqnarray}
C_l &:=& \left\{f(P) \mid f = \sum_{i=0}^m f_i X^i \in \tilde R[X], m \in
\NN, f_0 \neq 
0\right\}, \nonumber \\
C_r &:=& \left\{f(Q) \mid f = \sum_{i=0}^m f_i X^i \in \tilde R[X], m \in \NN, f_0 \neq 0\right\}, \label{eq:commutingSubsets}
\end{eqnarray}
where $\tilde R$ is the subring of constants of $S$, and $\tilde R[X]$ is
the univariate commutative polynomial ring over $\tilde R$. For an
element $f \in \tilde R[X]$, we let $f(P)$ denote the
substitution of $X$ in the terms of $f$ by $P$, and similarly
$f(Q)$ denotes the substitution of $X$ by $Q$. By this
construction, all the elements in $\mathcal{C}_l$ commute,
as do the elements in $\mathcal{C}_r$. The choice of the coefficient $f_0$ in both sets to be non-zero is motivated by the
following fact: If $f_0$ is allowed to be zero, Eve could find
that out by simply trying to divide the resulting
polynomial by $P$ on the left (resp. by $Q$ on the right). Moreover, Eve could iterate this process for
increasing indices, until an $f_i$ for $i \in \{0,\ldots, m\}$ is
reached, which is not equal to zero. This could lead to a decrease of
the amount of coefficients Eve has to figure out for certain choices
of keys. By the
additional condition of
having $f_0\neq 0$, Eve cannot retrieve any further information in the
described way.

Using this technique, the first steps of Algorithm
\ref{alg:keyExchange} can be altered in the following way. In step
\ref{stp:public}, $A$ and $B$ agree upon two elements $P,Q \in S$,
which represent $C_l$ and $C_r$, respectively, as in
(\ref{eq:commutingSubsets}). Then each one of them chooses two random
polynomials $f,g$ in
$\tilde R[X]$, and obtains the tuple of secret keys
in steps \ref{stp:secretA} and \ref{stp:secretB}
by computing $f(P)$ and $g(Q)$.  (Note that it could happen that
one, or even both, of $f(P)$, $g(Q)$ commutes with $L$, even though
neither $P$ nor $Q$ do so.  This appears to be unlikely in practice,
but in any case, it is straightforward to deal with this possibility.
$A$ (resp. $B$)
simply checks if the chosen $f_{A}(P)$ or $g_{A}(Q)$ (resp.
$f_{B}(P)$ or $g_{B}(Q)$) commutes with $L$.  If a commutation
with $L$ should occur, say for $A$, $A$ just chooses a new polynomial for $f_{A}(X)$
or $g_{A}(X)$.   As $P$ and $Q$ are chosen to not commute with $L$, $A$ will quickly find
polynomials $f_{A}(P)$ and $g_{A}(Q)$ that do not commute with $L$.)

\begin{example}
  \label{ex:ProtocolOnWeyl}
 Let $S$ be the third Weyl algebra $A_3$ over the finite field $\mathbb{F}_{71}$, upon which $A$ and $B$
 agree. Let
\begin{eqnarray*}
L &:=&3x_2^2 - 5\partial_2^2 - x_2\partial_3 - x_3 - \partial_2,\\
P &:=&-5x_3^2 - 2x_1\partial_3 + 34, \text{\ \ and }\\
Q &:=&x_2^2 + x_1x_3 - \partial_3^2 + \partial_3,
\end{eqnarray*}
where $L$ is the public polynomial as required in Algorithm 1, and
$P,Q$, such that they define the sets $\mathcal{C}_l$ and $\mathcal{C}_r$ as
in (\ref{eq:commutingSubsets}).

Suppose $A$ chooses polynomials $f_{A}(X) = 48X^2 + 22X + 27$,
$g_{A}(X) = 58X^2 + 5X + 52$, while $B$ chooses
$f_{B}(X) = 3X^2 + X + 31$, $g_{B}(X) = 24X^2 + 4X + 11$.
Then the tuples are 
$(P_A,\, Q_A) = (48P^2 + 22P + 27, \, \, 58Q^2 + 5Q + 52)$,
and $(P_B,\, Q_B) = (3P^2 + P + 31, \, \, 24Q^2 + 4Q + 11)$.
%Let $A$ and $B$ then choose as tuples $(P_A = 48P^2 + 22P + 27,
%Q_A = 58Q^2 + 5Q + 52)$ resp. $(P_B= 3P^2 + P + 31, Q_B=24Q^2 + 4Q + 11)$.

As described in the protocol, $A$ subsequently sends the product
$A_{\text{part}} := P_A \cdot L \cdot Q_A$ to $B$, while $B$ sends
$B_{\text{part}} := P_B \cdot  L \cdot Q_B$ to $A$, and their
secret key is
$P_A \cdot P_B \cdot
    L \cdot Q_B \cdot Q_A = P_B \cdot
    P_A \cdot L \cdot Q_A \cdot Q_B$. (For brevity, the
    final expanded product is not shown here.)

% To generate them, we used the computer algebra
% system \textsc{Singular}, which took less than a second to calculate
% all steps in the algorithm.
\end{example}

\begin{remark}
  For practical purposes, the degree of $L$ should be chosen to be of a
  sufficiently large degree in order to perturb the product $Q_B\cdot
  Q_A$ well enough before it is multiplied to $P_A\cdot P_B$. An
  examination of the best choices for the degree of $L$ is a subject of
  future work that includes practical applications of our primitive
  for a Diffie-Hellman-like key exchange protocol.
\end{remark}

\begin{remark}
As we will see in section \ref{sctn:insecurekeys}, there are known
insecure choices of keys for certain rings $S$. Obviously, in a
practical implementation, one has to check for these and avoid them.
\end{remark}

\subsection{Complexity of the Protocol}

Of course, as our definition of the rings we consider in  Algorithm
\ref{alg:keyExchange} --
namely rings of type (\ref{eq:ourRings}) -- is chosen to be as general
as possible, a complexity discussion is highly dependent on the choice
of the specific algebra. In practice, we envision that a certain
finite subset
of those algebras (such as, for example, the Weyl algebras, or
iterated extensions
of the rings used in \cite{boucher2010key}) will be studied for
practical applications. Our complexity discussion here
focusses rather on the general setup than on concrete examples.

As we generally assume, all arithmetics in $R$, and therefore also in
its subring of constants $\tilde R$, can be computed in polynomial
time. We suppose the same holds for the application of $\sigma_i$
and $\delta_i$,
for $i \in \{1, \ldots, n\}$, to the elements of $R$, and that the
time needed to choose a random element in $R$ is polynomial in the
desired bit length of this random element. Thus, the choice of
a random element in $S$
is just a finitely iterated application of the choice of coefficients,
which lie in $R$. Let $\omega_i(k)$ denote the cost of applying $\sigma_i$
(or $\delta_i$, depending on which one of them is non-trivial) to an
element $f \in R$ of bit-length $k \in \NN$. For two elements $f,g \in R$ of
bit-sizes $k_1,k_2 \in \NN$, we denote the cost of multiplying them in
$R$ by $\theta(k_1,k_2)$, and the cost of adding them by $\rho(k_1,k_2)$.

For the key exchange protocol the main cost that we need
to address is the cost of multiplying two polynomials in $S$. For a
monomial $m := \partial_1^{e_1} \cdots \partial_n^{e_n}$, where $e \in
\NN_0^n$, one can generalize Lemma \ref{lem:multCoeffWithMonom} to the
multivariate case and
find that multiplying $m$ and $f$, where $f$ has bit-size $k$,
costs $O(\prod_{i=1}^ne_i \cdot \omega_i(k))$ bit-operations. For
general polynomials in $S$, we obtain therefore the following
property:

\begin{lemma}
\label{lem:complexityMult}
Let $n$ be the number of Ore extensions as in (\ref{eq:ourRings}).
For two polynomials $h_1, h_2 \in S$, let $d \in \NN_0$ be the maximal
degree among the $\partial_i$ that appears in $h_1$ and $h_2$,
and let $k_1,k_2\in
\NN$ be the maximal bit-length among the coefficients of $h_1$
and $h_2$, respectively. For notational convenience,
we define $\zeta := \prod_{i=1}^{n} \omega_i(k_2)$. Then the
cost of computing the product $h_1 \cdot h_2$ is in 
\[
O\left(d^{2n}\cdot
\zeta \cdot \theta(k_1,\zeta)
\cdot \rho(\theta(k_1,\zeta),\theta(k_1,\zeta))\right).
\]
\end{lemma}

\begin{proof}
  We have at most $d^n$ terms in $h_1$. When we multiply $h_1$ and
  $h_2$, we have to regard each term separately, and compute the
  non-commutative relations. This results in the $d^{2n}$ different
  computations of size $\zeta$. Then, for every one of those results,
  we need to apply a multiplication in $R$ with the coefficients of
  $h_1$. In the end, the results of those multiplications have to be
  added together appropriately, which results in the above complexity.
\end{proof}

This lemma shows that multiplying two elements in
$S$ has polynomial time complexity in the size of the elements, since the value
of $n$ is fixed for a chosen $S$.

\begin{remark}
The cost in Lemma \ref{lem:complexityMult} assumes the worst case,
where each Ore extension of $R$ has a non-trivial $\delta_i$. If for
one of the extensions, $\delta_i$ is equal to the zero map, then the worst case
complexity in this variable is lower, as the term-wise multiplication does not
result each time in a sum of different terms in $\partial_i$. One can see here,
that in general, when the cost of the protocol
is crucial for a resource-limited practical implementation, one should prefer
Ore extensions where $\delta$ is the zero map, i.e. skew-polynomial rings.
\end{remark}

\subsection{Security Analysis}

\subsubsection{The Attacker's Problem}

The security of our scheme relies on the difficulty of the following
problem, which is similar to the computational Diffie Hellman problem
(CDH) \cite{maurer1994towards}.

Given a ring $S$, a security parameter $\nu$, two sets $\mathcal{C}_l,\mathcal{C}_r$ of multiplicatively
    closed, commutative subsets of $S$, whose elements do not commute
    with a certain given $L \in S$. Furthermore, let the products
    $P_A\cdot L \cdot Q_A$ and $P_B \cdot L \cdot Q_B$ for some
    $(P_A,Q_A)$, $(P_B,Q_B) \in \mathcal{C}_l \times\mathcal{C}_r$
    also be known.
{\quote \textbf{Difficult Problem (Ore Diffie Hellman (ODH)):} Compute $P_B\cdot P_A \cdot L
  \cdot Q_A \cdot Q_B\  (= P_A\cdot P_B \cdot L
  \cdot Q_B \cdot Q_A)$ with the given information.}\\

One way to solve this problem would be to recover one of the elements  $P_A,
Q_A,P_B$ or $Q_B$. This can be done via factoring $P_A\cdot L \cdot Q_A$ or $P_B \cdot L \cdot Q_B$
which appears, as mentioned in the introduction, to be hard. % .
% For some
% choices of $S$, factorization algorithms are known, but none of them
% runs in polynomial time. For others, as stated in the introduction, there
% is as yet no known algorithm, and for the non-Noetherian ones, even
% the existence of a factoring algorithm is questionable.
Furthermore, even if one is able to factor an intercepted product, the factorization may not
be the correct one due to the non-uniqueness of the factorization in Ore
extensions.

Another attack for the potential eavesdropper is to guess the degrees
of $(P_A,Q_A)$ (or $(P_B,Q_B)$) and to create an ansatz with the
coefficients as unknowns, to form a system of multivariate polynomial
equations to solve. This type of attack and its infeasibility
was discussed already in \cite{boucher2010key}, Section 5.2., and the
argumentation of the authors translates analogously to our setup.
Finally, another attempt, which seems natural, is to generalize
the attack of Dubois and Kammerer to the multivariate setup. %seems natural.
We will discuss such a possible
generalization for certain rings of type (\ref{eq:ourRings}) and show that it is
impractical in the following subsection.

We are not aware of any other way to obtain the common key of $A$ and $B$ while
eavesdropping on their
communication channel in Algorithm \ref{alg:keyExchange} other than trying
to recover the correct factorization from the exchanged products of
the form $P\cdot L \cdot Q$.

\begin{remark}
\label{rem:bruteForceAnsatz}
  Concerning the attack where Eve forms an ansatz and tries to solve
  multivariate polynomial systems of equations: In fact, each element
  in our system has total degree at most two. There exist attempts to
  improve the Gr\"obner computations for these kinds of systems
  \cite{courtois2000efficient,kipnis1999cryptanalysis}, but the
  assumptions are quite restrictive. % They e.g. assume that the systems
  % of equations are zero-dimensional, which is in general not the case
  % given an ansatz coming from our protocol.
  Besides the assumption that the given ideal must be zero-dimensional
  (which is only guaranteed in the case when the subring of constants is
  finite), there are certain relations between the number of
  generators and variables necessary to apply these improvements.
  We are not aware of any further progress on the techniques presented in
  \cite{courtois2000efficient,kipnis1999cryptanalysis} since 2000, which have
  fewer restrictions on the system to be solved.
\end{remark}

\begin{remark}
  Note that there is a corresponding decision problem related to ODH:
  Given a candidate for the final secret key, determine if this key is
  consistent with the public information exchanged by Alice and
  Bob. To the best of our knowledge, this is also currently intractable.
\end{remark}

\subsubsection{Generalization of the attack by Dubois and Kammerer}
\label{sbsbsctn:duboiskammerer}

In this subsection, we assume that our ring $S$ is Noetherian, and
that there exists a notion of a left or right Gr\"obner basis. Alice and Bob have
applied Algorithm \ref{alg:keyExchange} and their communication
channel has been eavesdropped by Eve. Now Eve knows about the chosen
ring $S$, the commuting subsets $\mathcal{C}_l,\mathcal{C}_r$ and the
exchanged products $P_A\cdot L \cdot Q_A$ and $P_B \cdot L \cdot Q_B$ for some
$(P_A,Q_A)$, $(P_B,Q_B) \in \mathcal{C}_l \times\mathcal{C}_r$. Let us
assume without loss of generality that Eve wants to compute $Q_A$.

Eve does not have a way to compute greatest common right divisors, but she can
utilize Gr\"obner basis theory. For this, she picks a finite family $\{E_i\}_{i=1}^m$, $m \in \NN$, of
elements in $C_r$. After that, she computes the set $G:=\{P_A\cdot L\cdot
Q_A\cdot E_i \mid i\in \{1, \ldots, n\}\}$.

All elements in $G$ (along with $P_A\cdot L \cdot Q_A$) have
$Q_A$ as a right divisor in common, since $E_i$ commutes with $Q_A$ for all
$i \in \{1, \ldots, m\}$. This means, the left ideal $I$ in $S$ generated by
the elements in $G \cup \left\{P_A\cdot L \cdot Q_A \right\}$ lies in -- or is even
equal to -- the left ideal
generated by $Q_A$. Hence, a Gr\"obner basis computation of $I$ might
reveal $Q_A$. If not, a set of polynomials of possibly smaller degree than the ones
given in $G$ that have $Q_A$ as a right divisor will be the result of
such a computation.

Besides having no guarantee that Eve obtains $Q_A$ from the
computations described above, the computation of a Gr\"obner basis is
an exponential space hard problem \cite{mayr1982complexity}. We tried
to attack our protocol using this idea. We chose the second Weyl
algebra as a possible ring, as there is a notion of a Gr\"obner basis
and there are implementations available. It turned out that
our computer ran out of memory after days of computation on several
examples where $L$, $Q_A$,
$Q_B$, $P_A$ and $P_B$ each exceed a total degree of ten. For practical
choices, of course, one must choose degrees which are higher (dependent on the
choice of the ring $S$). Hence,
we consider our proposal secure from this generalization of the attack
by Dubois and Kammerer.

\subsubsection{Recommended Key Lengths}

The question of recommended key lengths has to be discussed for each
ring of type (\ref{eq:ourRings}) separately. With lengths, one means in
the context of this paper the degree of the chosen public polynomials
$L$, $P$
and $Q$ in the $\partial_i$ for $i \in \{1,\ldots, n\}$ and the size
of their respective coefficients in $R$. We cannot state a general
recommendation for key-lengths that lead to secure keys for arbitrary
choices of $S$. For the Weyl algebra, where some implementations of
factoring algorithms
are available, we could observe through experiments that generic
choices of $P$ and $Q$ in $\mathcal{C}_l$ and $\mathcal{C}_r$,
respectively,
each of total degree 20, lead to products $P\cdot L \cdot Q$ which
cannot be factored after a feasible amount of time. %  This choice of
% degree also suggests
% that in the case where one uses a
% finite field as the ring of constants, a brute-force attack becomes
% impractical.
If one chooses our approach (\ref{eq:commutingSubsets})
to find commuting elements, the choice of the degree of the
polynomials in $\tilde R[X]$ is the critical part, and the polynomials
$P$ and $Q$ -- as they are publicly known -- can be chosen to be of
small degree for performance's
sake. 

In general, for efficiency, we recommend choosing $n=2$ for
the ring $S$, as it already
ensures that $S$ is not a principal ideal domain and keeps
multiplication costs low.

For the case where our underlying ring is a finite field, we are able to
present in Table \ref{tbl:bruteForceVsComputation} a more detailed
cost estimate on the hardness to attack our
protocol by using brute-force. There, we assume that $R = \FF_q$, where $q = p^k$ for a prime $p$
and $k>2$. For efficiency, as outlined above, we pick $n=2$ and
further $k=3$. Then we
define $S$ as being  $R[\partial_1;\sigma_1][\partial_{2};\sigma_2]$, where
$\sigma_1,\sigma_2$ are different powers of the Frobenius automorphism
on $\FF_q$. We assume that the polynomials are stored in dense
representation in memory. The two commuting subsets $\mathcal{C}_l, \mathcal{C}_r$
are defined as in \eqref{eq:commutingSubsets}.
We will measure the time in computation steps. We assume that any
arithmetic operation on $\FF_q$, as well as the application of $\sigma_1$
resp. $\sigma_2$, takes one step. Then, the cost formula as presented
in Lemma \ref{lem:complexityMult} will be in the worst case 
$d^4\cdot 2$, as addition and multiplication are assumed to take one
computation step, and $\zeta \leq 3$ (due to the automorphism
group of $R$ having order 3).
The security parameter is given as a tuple $(d_L, d_{PQ},\nu)$, where
$d_L$ is the total degree of $L$, $d_{PQ}$ is the total degree of each
of $P$
and $Q$, and $\nu$ is the maximal degree of the polynomials in
$\FF_p[X]$ chosen to compute each of $P_A, P_B, Q_A$ and $Q_B$. To
simplify the analysis, we assume for our estimates that the degree in each
$\partial_1$ and $\partial_2$ will be half of the total degree for $L,
P$ and $Q$.
As for the cost of Alice resp. Bob to compute the messages they
are sending, and to compute the final key, we used the following formulas to make a prudent
estimation:
\begin{itemize}
\item \textbf{Computing all powers of $P$ and $Q$:} The cost $c(\nu)$ to
    compute all these powers up to a certain exponent $\nu$ can be estimated by the following recursive formula:
    \begin{eqnarray*}c(1) &=& 0,\text{ (P or Q are given, no need to
      compute)}\\
    c(j+1) &=& \frac{(j\cdot d_{PQ})^4}{8} + c(j).
    \end{eqnarray*}
    As a closed formula, we can we can write it as
    $$c(\nu) = \sum_{j = 0}^\nu \frac{(j\cdot d_{PQ})^4}{8}= \frac{d_{PQ}^4}{8} \cdot \sum_{j = 0}^\nu j^4 =\frac{d_{PQ}^4}{8} \cdot \left(5\cdot
    (\nu+1)^5-\frac{1}{2}\cdot (\nu+1)^4+\frac{1}{3}\cdot (\nu+1)^3-\frac{1}{30}\cdot
    \nu-\frac{1}{30}\right).$$
\item \textbf{Generating private polynomials:} Both Alice and Bob have
  to compute $P_A$ and $Q_A$ resp. $P_B$ and $Q_B$. In order to do so,
  each power of $P$ and $Q$ has to be computed and multiplied by an element in
  $\FF_p$, which results in 
  $$2 \cdot \sum_{j =0}^\nu \frac{j^2\cdot
    d_{PQ}^2}{4} = \sum_{j =0}^\nu \frac{j^2\cdot d_{PQ}^2}{2} = \frac{d_{PQ}^2}{2}\cdot 3\cdot
    \left((\nu+1)^3-\frac{1}{2}\cdot (\nu+1)^2+\frac{1}{6}\cdot \nu+\frac{1}{6}\right)$$ operations.  Adding all
    these together adds another $2\cdot \frac{\nu^2\cdot d_{PQ}^2}{4}$
    operations for Alice resp. Bob.
  \item \textbf{Computing initial message:} We assume that we have the
  private polynomials for $A$ and $B$ already computed, and their respective degree is
  $d_{PQ}\cdot \nu$. Then, in order to compute the initial message,
  we need $\frac{(d_{PQ} \cdot \nu)^4}{8}$ steps to compute $P_A \cdot L$,
  assuming that the degree of $L$ is smaller. Afterwards, to obtain
  $P_A \cdot L \cdot Q_A$, we have to do $\frac{(d_{PQ} \cdot \nu +
  \frac{d_L}{2})^4}{8}$ additional steps.
\item \textbf{Computing the shared secret key:} The shared secret takes then $\frac{(2 \cdot d_{PQ} \cdot \nu +
  d_L/2)^4}{8} + \frac{(3 \cdot d_{PQ} \cdot \nu +
  d_L/2)^4}{8}$ steps to compute by directly applying the cost estimate for multiplication.
\end{itemize}
The worst case for the size in bits of the shared key in the end can be
estimated by adding the degrees of the computed $L$, $P_A$, $Q_A$, $Q_B$ and $P_B$ together. This results in
the formula
$$\left(\frac{d_L + 8\cdot \nu \cdot d_{PQ}}{2}\right)^2\cdot \lceil \log(p)\rceil,$$
where we assume that the partial degree in $\partial_1$ and $\partial_2$ is
about half of the total degree of the final polynomial. In practice,
one would probably prefer to use a sparse representation, which would
on average lead to smaller final key sizes.

As for the cost for an attacker to do a brute-force attack, i.e. trying to
determine the shared secret of Alice and Bob, we assume that an
attacker would try all possibilities for one of the polynomials $P_A,
P_B, Q_A$ or $Q_B$ and check, for each possibility, if the computed
polynomial divides one of the messages between Alice and Bob. Hence,
for every possibility, Eve must solve a linear system of
  equations of size $d_m^2$, where $d_m$ is the maximal total degree of
  one of the messages (usually $2\cdot \nu\cdot d_{PQ}+d_L$). I.e. there arise $(\frac{d_m}{2})^{2\omega}$ additional
  computation steps for each possibility, where $\omega$ is the matrix
  multiplication constant (currently $\omega \approx
  2.373$). Initially, the attacker has to also
  compute all powers of $P$ resp. $Q$, and then the additions, as listed
  above.
The following table lists our computed costs for different security
parameters.
\begin{center}
\begin{table}
\begin{tabular}{|c|c|c|c|c|c|}
\hline
\multirow{2}{*}{\textbf{Security
  Tuple}}&\multicolumn{3}{|c|}{\textbf{Computation Costs for Alice and
    Bob}} &\multirow{2}{*}{$\left. \begin{array}{c} \textbf{Final Key} \\
    \textbf{Size in KB} \end{array} \right.$
}&\multirow{2}{*}{$\left. \begin{array}{c} \textbf{Brute-Force}\\
      \textbf{Cost} \end{array} \right.$}\\
\cline{2-4}
& \textbf{Secret Parameter} & \textbf{Initial Message} &
 \textbf{Shared Secret} && \\
\hline
\hline
$(30,5,10)$ & $1.247955E+08$ & $3.012579E+06$ & $1.145127E+08$ &
$46$   & $2.066009E+16$\\
$(30,5,15)$ & $8.144450E+08$ & $1.215633E+07$ & $5.073701E+08$ & $97 $  & $9.616857E+18$\\
$(30,5,20)$ & $3.176336E+09$ & $3.436258E+07$ & $1.497794E+09$ & $169$  & $2.237607E+21$\\
$(30,5,25)$ & $9.248193E+09$ & $7.853758E+07$ & $3.508245E+09$ & $260$  & $3.467317E+23$\\
$(30,5,30)$ & $2.229704E+10$ & $1.559313E+08$ & $7.074856E+09$ & $370$  & $4.110897E+25$\\
$(50,5,35)$ & $4.711215E+10$ & $3.172363E+08$ & $1.391021E+10$ & $514$  & $5.144021E+27$\\
$(50,5,40)$ & $9.029806E+10$ & $5.203613E+08$ & $2.315166E+10$ & $665$  & $4.258507E+29$\\
$(50,5,45)$ & $1.605675E+11$ & $8.086426E+08$ & $3.637583E+10$ & $836$  & $3.176783E+31$\\
$(50,5,50)$ & $2.690343E+11$ & $1.203174E+09$ & $5.458994E+10$ & $1027$& $2.179949E+33$\\
\hline
\end{tabular}
\caption{Computation costs (given as number of primitive computation steps) for Alice and Bob to perform Algorithm \ref{alg:keyExchange}
  with $R = \FF_q$, $p = 2$ and $S =
  R[\partial_1;\sigma_1][\partial_2;\sigma_2]$ and costs for Eve to
  perform a brute-force attack.}
\label{tbl:bruteForceVsComputation}
\end{table}
\end{center}

%  We could see in Example
% \ref{ex:ProtocolOnWeyl} that for seemingly small choices for $L$, $P$,
% $Q$, $(P_A, P_B)$ and $(Q_A, Q_B)$, one obtains a large key.

\begin{remark}
We tried
to factor the exchanged products $P_A\cdot L \cdot Q_A$ and $P_A \cdot
L \cdot Q_B$ from the small Example \ref{ex:ProtocolOnWeyl} 
in section \ref{subsec:description-of-the-protocol} using
\textsc{Singular} and \textsc{REDUCE}, and it turned out that both
were not able to provide us with one factorization after 48 hours of
computation on an iMac with 2.8Ghz (4 cores) and 8GB RAM available.
This means that even for rather small choices of keys, the
recovery of $P$ and $Q$ via factoring appears already to be hard using available
tools. Of course, for this small key-choice, a
brute-force ansatz attack (as described above in Remark \ref{rem:bruteForceAnsatz}) would succeed fairly quickly. We also tried
150 examples with different degrees for $P, Q, L$ and the respective
polynomials in $C_l$ and $C_r$. In particular, we let the degree of $L$ range from
$5$ to degree $50$, the degree of $P$ and $Q$ respectively between $5$
and $10$, and the degrees of the elements in $C_l$ and $C_r$ --- which
are created with the help of $P$ and $Q$ -- are having degrees ranging
between $25$ and $50$. We gave each factorization process a time limit
of 4 hours to be finished. None of the polynomials has been factored
within that time-frame. The examples can be downloaded from the
following website:
\url{https://cs.uwaterloo.ca/~aheinle/software_projects.html}.
\end{remark}

\subsubsection{Attacks On Similar Systems}

Here, we discuss why known attacks on protocols similar to Algorithm \ref{alg:keyExchange} do not
apply to our contexts.

As emphasized before, the attack developed by Dubois and Kammerer on the
protocol by Boucher et al. is prevented by choosing rings that are
not principal ideal domains. Thus, there is no general algorithmic way to compute
greatest common right divisors.

When applying the rings $S$ of type (\ref{eq:ourRings}) to exchange
keys, one does in fact not utilize the whole ring structure, but only the
multiplicative monoid structure. Therefore it appears to be reasonable
to consider also attacks developed for protocols based on non-abelian
groups (albeit they contain more structure than just monoids, the
latter being the correct description of our setup). The most famous
protocol is given by Ko et al., as discussed in the section on related
work. The attack developed by Jun and Cheon exploits the fact that
braid groups are linear. However, there is currently no linear representation
known for our rings of type (\ref{eq:ourRings}) (though it would be an
interesting subject of future research), so there is at present no analogous
attack on protocols based on our primitive.  Furthermore, even if a linear
representation for our rings were discovered, it is not clear whether Jun and
Cheon's attack could be extended to our case, as the authors
make use of invertible elements in their algorithm (which our structures,
only being monoids, do not possess).

\section{Implementation and Experiments}
\label{sctn:impl}

We developed
an experimental implementation of the
key exchange protocol as presented in Algorithm
\ref{alg:keyExchange} in the programming language
\textsc{C}\footnote{One can download the implementation
at \url{https://github.com/ioah86/diffieHellmanNonCommutative}}. We
decided to develop such a low-level implementation after we found
that commodity computer algebra systems appear to be too slow to make
experiments with reasonably large elements. This may be due to the
fact that their implemented algorithms are designed to be generally applicable
to several classes of rings and therefore come with a large amount of computational overhead.
Our goal is to examine key-lengths and the time it takes
for computing the secret keys. It is to be emphasized that our code
leaves considerable room for improvement.

For the implementation, we chose our ring $S$ to have the form as described
in Example \ref{ex:boucherExtended}.
In particular, our ring for the coefficients $R$ is set to
$\FF_{125}$, and we fixed $n := 2$. Internally, we view
$\FF_{125}$ isomorphically as
$\FF_5(\alpha) :=\FF_5[x]/\langle x^3 + 3x + 3\rangle$. Our
non-commutative polynomial ring $S$ is
$R[\partial_1;\sigma_1][\partial_2; \sigma_2]$, where
\begin{eqnarray*}
  \sigma_1&:& \FF_5(\alpha) \to \FF_5(\alpha),
    \text{\ }a_0 + a_1\alpha + a_2
  \alpha^2 \mapsto a_0 + a_1 + a_2 + 3a_2\alpha + (3a_1 +
  4a_2)\alpha^2\\
  \sigma_2&:& \FF_5(\alpha) \to \FF_5(\alpha),
    \text{\ }a_0 + a_1\alpha + a_2
  \alpha^2 \mapsto a_0 + 4a_1 + 3a_2 + (4a_1 + 2a_2)\alpha + 2a_1\alpha^2.
\end{eqnarray*}

 The ring of constants is
therefore $\tilde R :=\FF_5 \subset R$. These two automorphisms are
given by different powers of the Frobenius automorphism, and they are
the only two distinct non-trivial automorphisms on $\FF_{5}(\alpha)$
(cf. \cite[Theorem 12.4]{garling1986course}).
% We obtained these two different automorphisms on $\FF_{5}(\alpha)$
% using the computer
% algebra system \textsc{Sage} \cite{sage}.

Note, that the multiplication of two elements
$f$ and $g$ in this ring requires $O(n^4)$ integer multiplications, where $n = \max\{\deg(f), \deg(g)\}$.

Following the notation as in Algorithm \ref{alg:keyExchange}, our
implementation generates random polynomials $L$, $P$ and $Q$ in
$S$. Our element $L$ is chosen to have total degree 50, and $P$, $Q$
each have total degree 5. Afterwards, it generates four polynomials in $\tilde R[X]$ to obtain
$(P_A,Q_A), (P_B,Q_B)$ in the fashion of (\ref{eq:commutingSubsets}).

Subsequently, the program computes the products $P_A \cdot L \cdot
Q_A$, $P_B \cdot L \cdot Q_B$ and the secret key $P_A\cdot P_B \cdot L
\cdot Q_B \cdot Q_A = P_B\cdot P_A \cdot L
\cdot Q_A \cdot Q_B$. Naturally, some of those computations would be
performed in parallel when the protocol is applied, but we did not
incorporate parallelism into our experimental setup. At runtime, all
computed values are printed out to the user.

We experimented with different degrees for the polynomials in $\tilde
R[X]$ to generate the private keys, namely 10, 20, 30, 40 and 50. This
leads to respectively 20, 40, 60, 80 and 100 indeterminates for 
Eve to solve for if she eavesdrops the channel between Alice and
Bob and tries to attack the protocol using an ansatz by viewing the
coefficients as unknown parameters. Even if she decides to attack the protocol using brute-force, she
has to go through $5^{10}, 5^{20}, 5^{30}, 5^{40}$ and $5^{50}$
possibilities respectively (note here, that for a brute-force attack,
Eve only needs to extract a right or left hand factor of the products
$P_A \cdot L \cdot
Q_A$ and $P_B \cdot L \cdot Q_B$ that
Bob and Alice exchange).
The file sizes and the timings for the experiments are illustrated in
Figure \ref{fig:timings}.

\begin{figure}
%\includegraphics[width=.5\textwidth]{gnuplotTimings.pdf}
%\ \ \includegraphics[width=.5\textwidth]{gnuplotSizes.pdf}
\includegraphics[width=.5\textwidth]{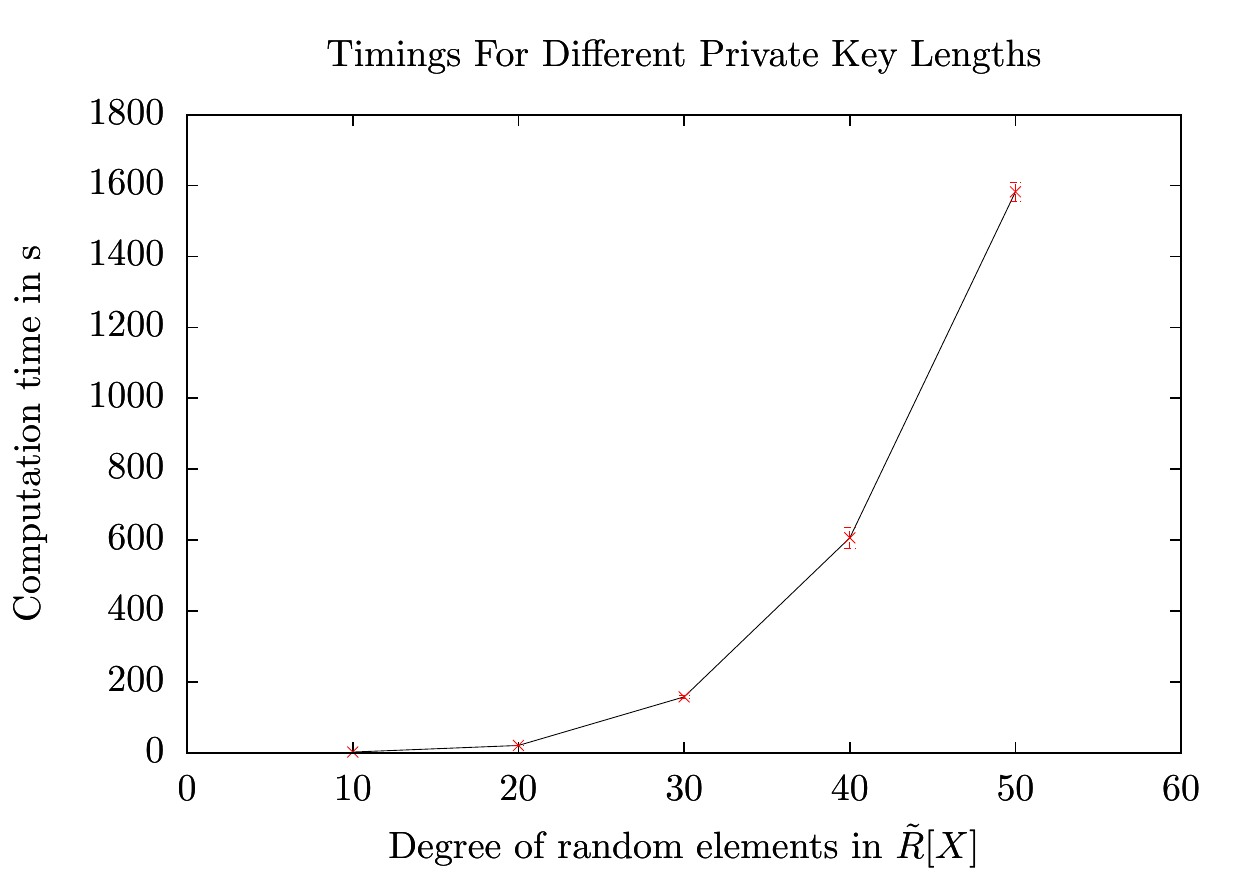}
\ \ \includegraphics[width=.5\textwidth]{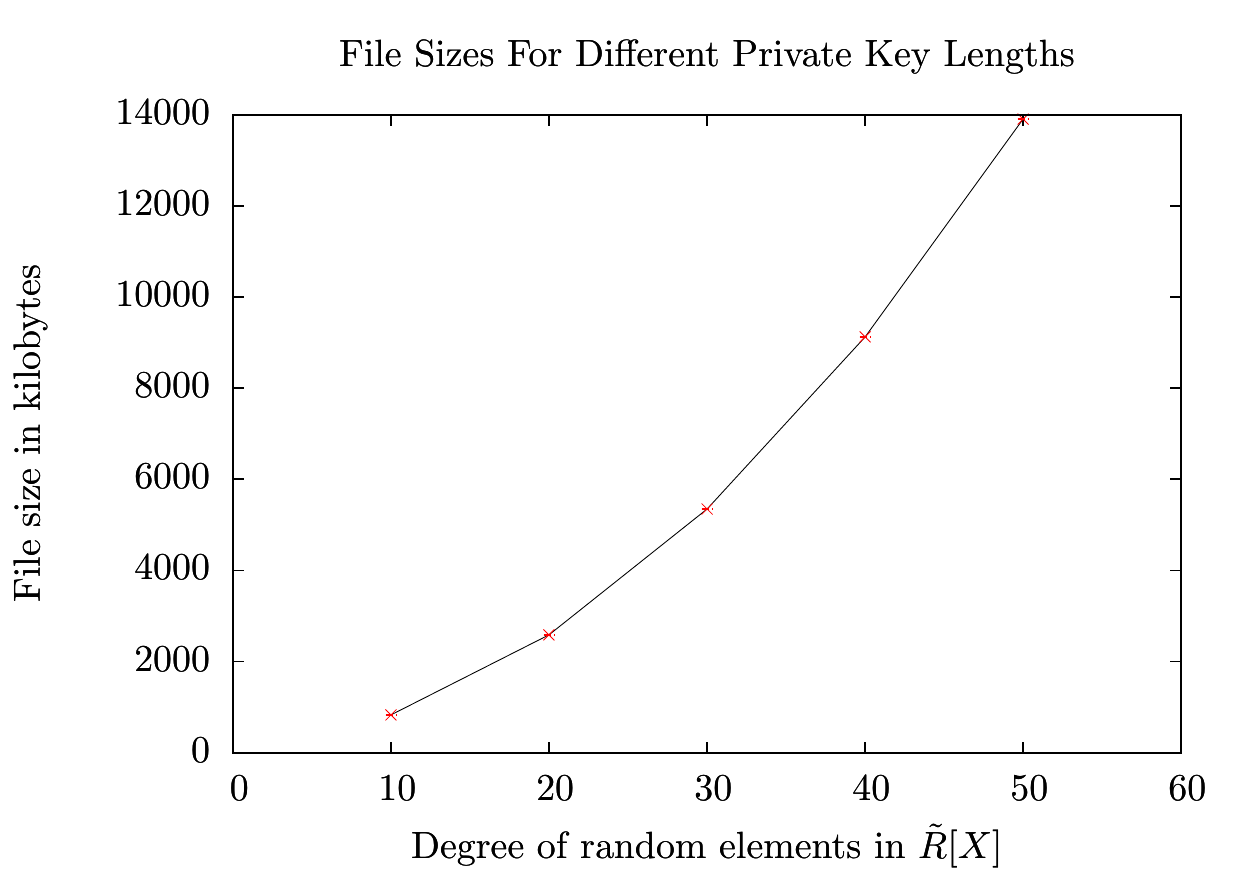}
\caption{Timings and file sizes for different degrees of elements in
  $\tilde R[X]$}
\label{fig:timings}
\end{figure}

Note, that the file sizes are not indicative of the actual bit-size of the
keys, as the files we produced are made to be human-readable. Allowing
for this fact, the bit-sizes of our keys are comparable to those found
necessary for secure implementations of the McEliece cryptosystem
\cite{mceliece1978public,bernstein2008attacking}, which is a
well-studied post-quantum encryption scheme.

In our experimental setup, we can see that one can generate
a reasonably secure key (degree 30 for the elements in
$\tilde R[X]$) in less than five
minutes at the current stage of the implementation.
For larger degrees, we believe that machine-optimized
code would decrease the computation time significantly. An interesting
question is whether arithmetics in our class of non-commutative rings
can be implemented in a smart way on a quantum computer.

\subsection{Challenge Problems}
\label{sbsctn:challenges}
For readers who would like to try to attack the keys generated by this particular
implementation, we have generated a set of challenge problems. They
can be found, with description, on the website of one of the authors
(\url{https://cs.uwaterloo.ca/~aheinle/miscellaneous.html#challenges}). 
There are also challenges included for the three-pass protocol as
described in section \ref{subsec:double-lock}.

\section{Insecure Keys}
\label{sctn:insecurekeys}
\begin{comment}
 - Homogeneous polynomials
 - Polynomials P and Q commuting with L
 - Polynomials in P resp. Q that commute with L. (hard to get,
   Schwarz-Zippel Argument)
 - Testing of key in the end (P_a Q_b do not commute with L)
 - The center has to be trivial! (Otherwise Giesbrecht's algorithm and
   I have and idea for an attack based on finding central elements)
\end{comment}

In this section, we will present an insecure key-choice for the Weyl
algebras. The construction of those insecure keys, which is related to
finding commutative subrings, can be applied to other
algebras of type (\ref{eq:ourRings}).

\subsection{Insecure Keys for the Weyl Algebras}

\subsubsection{Graded Polynomials.}
Based on the paper \cite{giesbrecht2014factoring},
there is a  large subset of the polynomial $n^{\mathrm{th}}$ Weyl algebra,
where the
factorization problem of their elements can be reduced to factoring
in a commutative, multivariate polynomial ring.

In particular, there exists a non-trivial $\ZZ^n$-grading on the
polynomial $n^{\mathrm{th}}$ Weyl algebra,
where the \textbf{$z^{\mathrm{th}}$ graded part}, for $z \in \ZZ^n$, can be
characterized in the following way:
$$A_{n}^{( z )} := \left\{ \sum_{e,w  \in \NN^{n}_{0} \atop w-e =z} r_{ew} \cdot
  x_1^{e_1}\cdots x_n^{e_n} \cdot \partial_1^{w_1} \cdots \partial_n^{w_n}
  \quad | \quad   r_{ew} \in \tilde \KK \right\}.$$

We call an element $h$ in the polynomial $n^{\mathrm{th}}$ Weyl algebra
\textbf{graded}, if $h \in A_n^{(z)}$ for some
$z \in \ZZ^n$. These graded polynomials are exactly the ones for which
the factorization problem can be reduced to commutative factorization
as mentioned above.

Now, there are two possible scenarios for weak key choices of $A$ and
$B$. Let $L$ be the public key, and $P_A,Q_A$
and $P_B, Q_B$ be the private keys of $A$ and $B$ respectively,
i.e. the final key of $A$ and $B$ is $P_A\cdot P_B\cdot L\cdot
Q_B\cdot Q_A$. The first scenario is that all of the keys that $A$ and $B$ use
are graded. Then a possible eavesdropper $E$ can recover the private
keys by factoring $P_A\cdot L\cdot Q_A$ and $P_B\cdot L\cdot Q_B$, applying techniques
presented by \cite{giesbrecht2014factoring}.
The second scenario is that one of the private keys of $P_A, P_B, Q_A$
and $Q_B$ is graded. Without loss of generality, let $P_A$ be
graded. Then $E$ can recover $P_A$ by factoring every graded summand
of $P_A\cdot L\cdot Q_A$, and therefore $E$ can also recover $Q_A$ and the
security is broken.

Fortunately, $A$ and $B$ can check their keys for being graded in
polynomial time. In particular, for an element $h$ the polynomial
$n^{\mathrm{th}}$ Weyl algebra, one has to check if for every $i \in \{1,\ldots, n\}$ the
difference of the exponents of $x_i$ and $\partial_i$ is the same in
every term of $h$. This is the case if and only if $h$ is graded.

\begin{remark}
  One can argue that the Weyl algebras as we define them assume that
  the 
  $x_i$ are units, and therefore the attack as described here is not
  possible once we choose for our keys coefficients that have  in the denominator some nontrivial
  polynomial in the $x_i$. But there is a
  possibility to lift factorizations into the polynomial Weyl algebra,
  which is described in \cite{giesbrecht2014factoring}.
  Thus one has to lift the keys and double check for them
  being graded or not. This check obviously still requires only polynomial time.
\end{remark}

% DEPRECATED: Attack not valid.
% \subsection{Insecure Keys for Any Choice Of Algebra}
% This subsection is dedicated to problems that occur independent of the
% concrete algebra $A$ and $B$ are choosing their keys from.
% \subsubsection{Accidentally Commuting Products of Private Keys}

% Let $R$ be the algebra of type (\ref{eq:ourRings}) that $A$ and $B$
% agree on. Furthermore, let $L$ be the public Key, and $P_A,Q_A$
% and $P_B, Q_B$ be the private keys of $A$ and $B$ respectively,
% i.e. the key of $A$ and $B$ is $P_AP_BLQ_AQ_B$. Assume that there is
% an eavesdropper $E$ and somehow $E$ got to know $R$, $L$, $P_BLQ_B$
% and $P_ALQ_A$. Now assume the following accidental conditions:
% \begin{itemize}
% \item The product $Q_BP_A$ does
%   commute with $L$.
% \item The product $P_BQ_B$ commutes with $L$.
% \item $P_B$ commutes with $L^2$.
% \end{itemize}
%  Then $E$ can
% build the product
% \[
%   P_BLQ_B P_ALQ_A = P_BL^2Q_BP_AQ_A = L^2P_BQ_BP_AQ_A = LP_BQ_BLP_AQ_A.
% \]

Generally, as we can see here, one should study the
concrete ring of type (\ref{eq:ourRings}) chosen for Algorithm
\ref{alg:keyExchange}, in order to determine in which cases
factorization can be reduced to an easy problem in commutative algebra, and avoid those
cases for the key-choice. The construction of the insecure keys in the
case of the Weyl algebra as presented here gives an idea how to find
commutative subrings and how to avoid them.

\section{Enhancements}
\label{sec:Enhancements}

In what follows, we show that the use of multivariate Ore polynomials
is not limited to the Diffie-Hellman-like protocol discussed
and analyzed in the previous sections, but can also be utilized to
develop other cryptographic applications.

It should be noted that two of the protocols described below, namely
the digital signature scheme and the zero-knowledge-proof protocol, do
not require the use of commuting subsets, which reduces the necessary
amount of public information to be exchanged between Alice and Bob.

\subsection{The Three-pass Protocol and private $L$}
\label{subsec:double-lock}

As we saw in Algorithm \ref{alg:keyExchange},
Alice and Bob must make inter alia the following information public at the start:
$\mathcal{C}_l$, $\mathcal{C}_r$, and $L$.  If we can find a way to make any of these
objects private between Alice and Bob, the security of the system
may be increased.
% Certainly, it cannot make the system any less secure, other things
% being equal.

One possible method to make $L$ private is based on a well-known
puzzle, known variously as the Locks and Boxes, the Knight and
the Princess, and perhaps other names.  In this puzzle, Alice wishes to
send Bob an item in a securely locked box.  Both have a supply of
locks and keys to the locks, which they could use to lock the box.
 However, neither
Alice nor Bob have keys to the other party's locks.  So if Alice
sends the item to Bob in a locked box, he will be unable to open it.
Furthermore, it does no good to send an unopened lock or a key
to the other party to use, as it would simply be stolen, or the key
copied, negating any security.
 
The solution is for Alice to first send Bob the item in a box,
sealed with one of her own locks.  Bob cannot open Alice's lock,
but he can add one of his own locks to the box (so it has now has
two locks on it) and return it to Alice.  Alice then removes her
lock, and sends the box a second time to Bob.  Finally, Bob
removes his lock, and opens the box.

We use this idea in the following protocol, for Alice to send
a secret choice of $L$ to Bob. %  As we have not found any standard
% name for it in the literature, we have tentatively called it
% the ``double-lock protocol":

\begin{algorithm}[H]
  \begin{algorithmic}[1]
    \STATE $A$ and $B$ publicly agree on a ring $S$ of type
    (\ref{eq:ourRings}), and two multiplicatively
    closed, commutative subsets $\mathcal{C}_l,\mathcal{C}_r
    \subset S$.% non-central elements $P,Q \in S$,
    % such that $P$ and $Q$ do not commute with each other. %\label{stp:public}
    \STATE $A$ chooses a secret $L \in S$, which is not central in
    $\mathcal{C}_l$ and $\mathcal{C}_r$, that she wants to share
               with $B$.
               % (Note: If $L$ commutes with either $P$ or $Q$,
               %  then either $A$ must choose a different secret
               %  $L$, or
               %  $A$ and $B$ must agree on a different public
               %  $P,Q$, with $A$ checking that neither commutes with $L$.)
    \STATE $A$ picks random polynomials $P_A \in \mathcal{C}_l$ and
    $Q_A \in \mathcal{C}_r$, which form her private tuple $(P_A, Q_A)$. If coincidentally either $P_A$ or
                $Q_A$ commute with $L$, $A$ must choose a different
                pair $P_A$, $Q_A$. \label{step:paqa}
    % \STATE $A$ checks that none of $P_{A}$, $Q_{A}$, $L$
    %            are pairwise commutative.  If any pair commutes,
    %            return to step \ref{step:paqa} and repeat it,
    %            until this non-commuting condition is satisfied.
    \STATE $A$ computes the product $P_{A}\cdot L\cdot Q_{A}$,
               and sends it to $B$.
    \STATE $B$ picks random polynomials $P_B \in \mathcal{C}_l$ and
    $Q_B \in \mathcal{C}_r$, which form his private tuple $(P_B, Q_B)$.
    \STATE $B$ computes the intermediate product
               $P_{\text{int}}
                = P_{B}\cdot P_{A}\cdot L\cdot Q_{A}\cdot Q_{B}(=P_{A}\cdot P_{B}\cdot L\cdot Q_{B}\cdot Q_{A})$
               and sends it to $A$.\\
               % (Note that by the commutativity of $P_{A}$ and $P_{B}$,
               % and of $Q_{A}$ and $Q_{B}$, we also have
               % $P_{\text{int}}
               %  = P_{A}\cdot P_{B}\cdot L\cdot Q_{B}\cdot Q_{A}$.)
    \STATE $A$ divides $P_{\text{int}}$ on the right by $Q_{A}$,
               and on the left by $P_{A}$.  % (There will be no
               % remainders.)
               $A$ sends the result,
               $P_{B}\cdot L\cdot Q_{B}$, to $B$.
    \STATE $B$ divides $P_{B}\cdot L\cdot Q_{B}$ on the right
               by $Q_{B}$, and on the left by $P_{B}$%  (again,
               % there will be no remainders)
               , to recover the
               secret $L$.
  \end{algorithmic}
\caption{Three-pass exchange protocol with rings of
         type (\ref{eq:ourRings})}
  \label{alg:doubleLock}
\end{algorithm}
\noindent
% \textit{Remark 1: } Note that if $S$ were a Euclidean domain,
% the eavesdropper could easily find $L$ (along with $P_{A}$,
% $Q_{A}$, $P_{B}$, $Q_{B}$)), using gcrd and gcld operations.
% But our choice of $S$ is immune to such attacks, even on a quantum
% computer, and hopefully
% is also safe from any practical factoring attack.\\

\begin{remark} With $L$ secretly agreed upon, Alice and Bob
may now use Algorithm \ref{alg:keyExchange} to agree upon a
secret key, with only the commuting elements from $\mathcal{C}_l$ and
$\mathcal{C}_r$ being public.  Note that they
may choose a new set of tuples $(P_{A},\text{\ }Q_{A})$,
$(P_{B},\text{\ }Q_{B})$, and indeed, may even publicly agree
on a new choice of $\mathcal{C}_l$ and $\mathcal{C}_r$ for Algorithm \ref{alg:keyExchange}.
In this way, the information exchanged during Algorithm
\ref{alg:doubleLock} cannot help the eavesdropper to know the
$(P_{A},\text{\ }Q_{A})$ and $(P_{B},\text{\ }Q_{B})$
used in Algorithm \ref{alg:keyExchange}. % At most, it may
% reveal information about $L$.  But if we used Algorithm
% \ref{alg:keyExchange} without this preprocessing step,
% $L$ would have been public anyways.  For this reason, we may say
% that using the double-lock protocol first, and then the
% Diffie-Hellman-like protocol, will always be at least as secure
% as using the Diffie-Hellman-like protocol alone.
\end{remark}

\begin{remark}
Naturally, Algorithm \ref{alg:doubleLock} can be used as a key
exchange protocol by itself, where $L$ is the key being
exchanged. Deciding which approach would be better in a given
situation would require further investigation. However, there is an
advantage in defending against the attack where Eve forms an ansatz and
tries to solve non-linear systems of equations: Eve would deal with
tertiary systems of equations (instead of quadratic ones) in this
case, as the coefficients of $L$ are also unknown.

 % Since the value of $L$ agreed upon at the
% end of Algorithm \ref{alg:doubleLock} is, hopefully, secret,
% why not use it directly as the final secret key, rather than
% going on to run Algorithm \ref{alg:keyExchange}?  This requires
% further investigation, but it is an interesting possibility.
% In particular, note that a brute-force ansatz attack may be more
% difficult to apply to Algorithm \ref{alg:doubleLock}.  For 
% Algorithm \ref{alg:keyExchange}, the only unknowns in an ansatz
% are the coefficients of $P_{A}=f_{A}(P)$, $Q_{A}=g_{A}(P)$.
% But for Algorithm \ref{alg:doubleLock}, we also have unknown
% coefficients
% in the structure of $L$.  Thus, for similar amounts of work,
% using Algorithm \ref{alg:doubleLock} on its own may, in fact,
% be more secure than using Algorithm \ref{alg:keyExchange}.
% But more investigation is needed.
\end{remark}

% \begin{remark}
% In principle, one could use
% Algorithm \ref{alg:doubleLock} to exchange an actual
% message directly, rather than to create a key which will then
% be used to
% encrypt the message.  To do this, Alice would need some injective,
% invertible method of converting a plaintext message into
% an operator $L \in S$.  Alice then sends $L$ to Bob using
% Algorithm \ref{alg:doubleLock}, who then converts $L$
% back into a plaintext message.

% However, this conversion seems to us to be overly complicated.
% It may be easier to simply use the method suggested by
% Boucher et al. \cite{boucher2010key}.  There, they suggest
% simply exchanging enough distinct secret keys until their
% total length is at least equal to the length of the message
% to be encrypted, and then encrypting the message using
% bitwise XOR between the keys and the message.
% \end{remark}

\begin{remark}
Regarding the requirement, at various points
in Algorithm \ref{alg:doubleLock}, that elements must not
commute: If the elements did commute, we actually do not know
of any attacks
that would take advantage of this property.  So these requirements
may be unnecessary.  However, in general terms, commuting algebraic
objects are often easier to analyze than non-commuting ones, and
might be easier to attack.  Prudence therefore suggests that we
choose non-commuting elements as much as possible, allowing
commutativity only where it is needed by the protocol in order to
function properly.
 
% Furthermore, as an alternative to the first few steps of
% Algorithm \ref{alg:doubleLock}: If circumstances permit,
% it may be better for Alice to first choose $L$, $P$, $Q$
% privately, and ensure they are pairwise non-commuting.
% Only then should she publish $P$ and $Q$ for Bob to see.
% This would be better than having to adjust $P$ and $Q$
% after they have been published, since, in principle, this
% public change would give an attacker some information about $L$.
% (Namely, that $L$ commutes with one of the initial published
% values of $P$ or $Q$.)\\
\end{remark}

\begin{comment} 
\textit{Remark 6: } An interesting question is whether
Algorithm \ref{alg:doubleLock} can be made more efficient
by making it ``one-sided".  That is, suppose Alice begins
with a secret $L$, and just one private key $P_{A}$.
Let her form the product $P_{A}\cdot L$ and send it to Bob.
Bob multiplies this on the left by his private key $P_{B}$,
to form $P_{B}\cdot P_{A}\cdot L = P_{A}\cdot P_{B}\cdot L$,
and sends this to Alice.  Alice divides on the left by $P_{A}$
to yield $P_{B}\cdot L$, and returns this to Bob.  Finally,
Bob divides on the left by $P_{B}$, recovering $L$.

Now, the attacker does not know $L$, and only the general
structure of $P_{A}$ or $P_{B}$ is public.  So once again,
the attacker's strategies appear to be limited to either
factoring, or setting up an ansatz and solving the resulting
nonlinear system of multivariate polynomial equations.
So this simpler protocol may also be secure.  Nonetheless,
prudence suggests we be cautious, and we therefore recommend
the version of Algorithm \ref{alg:doubleLock}, which uses
``two-sided" locks on $L$, multiplying $L$ by private keys on
both the left and right sides.

****Ooops, no... just divide Pb*Pa*L on the right by Pa*L,
    to get Pb, and so on.  Not a good idea, still need
    double-sided locking.

\end{comment}

The security assumption here differs slightly from the
one for Algorithm \ref{alg:keyExchange}. Here, Eve is given
$\mathcal{C}_l$, $\mathcal{C}_r$,
$P_A\cdot L \cdot Q_A$, $P_B\cdot P_A\cdot L \cdot Q_A\cdot Q_B$ and
$P_B\cdot L \cdot Q_B$.\\[5pt]
{\quote{\textbf{Difficult Problem (Ore Three Pass Protocol (OTPP)): }Determine $L$ with the given
  information.}}\\[5pt]
As is the case for the Diffie-Hellman-like protocol, the ability to feasibly
compute all factorizations of elements in $S$ would of course allow
Eve to determine Alice's secret. But this is a hard problem with
currently known methods, as
mentioned in the introduction.
We are not aware of any other feasible way to obtain $L$.

\begin{remark}
  Similarly to ODH, there is also a corresponding decision problem related to OTPP:
  Given a candidate for the final secret $L$, determine if $L$ is
  consistent with the public information exchanged by Alice and
  Bob. To the best of our knowledge, this is also currently intractable.
\end{remark}

\subsection{ElGamal-like Encryption and Signature Schemes}

In 1984, ElGamal showed that a cryptosystem and digital
signature scheme were possible using exponentiation in
finite fields \cite{elgamal1985public}.  Here, we will show that such schemes
are also possible using non-commutative polynomial rings.
There will necessarily be some differences between our
schemes and those of ElGamal, since finite fields are
commutative and all (non-zero) elements are invertible,
whereas our structures are non-commutative and the elements are
non-invertible.  Nevertheless, approaches very similar to
those of ElGamal can be developed.

\subsubsection{An ElGamal-like Encryption Scheme}
\label{subsubsec:ElGamal-encryption}

Suppose Bob wishes to send a secret message to Alice.
An inconvenience of the Diffie-Hellman-like protocols is that
Bob must wait for Alice to respond before the key is decided.
Only then can Bob encrypt and send his message.

It would be desirable to have a scheme whereby Bob, or anyone
else, can send an encrypted message to Alice whenever they
wish, without the need to wait for Alice to respond.  This
is possible, as we now show, but it requires Alice to precompute
and publish some information. It should be noted that the following
encryption scheme is only intended to provide a basic level of
security (i.e. an eavesdropper who intercepts the ciphertext is not able to
compute the corresponding plaintext). Further enhancements to the scheme are required for
stronger notions of security (e.g. indistinguishability against
adaptive chosen-ciphertext attacks).

\noindent
\textbf{Preparation:} Alice chooses $L\in S$, and two multiplicatively
    closed, commutative subsets $\mathcal{C}_l,\mathcal{C}_r
    \subset S$. Then she
 picks random polynomials $P_A \in \mathcal{C}_l$ and
    $Q_A \in \mathcal{C}_r$, which form her private tuple $(P_A, Q_A)$. If coincidentally either $P_A$ or
                $Q_A$ commute with $L$, Alice must choose a different
                pair $P_A$, $Q_A$.
Alice publishes $L$, $\mathcal{C}_l$, $\mathcal{C}_r$, and
$P_{\text{Alice}}:=P_{A}\cdot L \cdot Q_{A}$.
The tuple $(P_{A},\text{\ }Q_{A})$ is kept secret.

\noindent
\textbf{Encryption:} Let $m\in S$ be a message which Bob
wishes to encrypt before sending to Alice.  (We assume
that a plaintext message has already been mapped to
Ore polynomial form, so $m \in S$.)

\noindent
Bob picks random polynomials $P_B \in \mathcal{C}_l$ and
$Q_B \in \mathcal{C}_r$, which form his private tuple $(P_B,
Q_B)$. Again, Bob ensures that $P_B$, $Q_B$ and $L$ are pairwise
non-commutative.

\noindent
Bob computes $P_{\text{Bob}}:=P_{B}\cdot L \cdot Q_{B}$
and $P_{\text{final}}:= P_{B}\cdot P_{\text{Alice}}\cdot Q_{B}$,
and encrypts the message as $m_{e} := m\cdot P_{\text{final}}$.
Bob then sends the pair $\left(m_{e},P_{\text{Bob}}\right)$
to Alice.

\noindent
\textbf{Decryption:} Given $\left(m_{e},P_{\text{Bob}}\right)$,
Alice first computes
$P_{\text{final}} = P_{A}\cdot P_{\text{Bob}}\cdot Q_{A}$.
She then divides $m_{e} (= m\cdot P_{\text{final}})$ on the
right by $P_{\text{final}}$ to recover $m$.

\noindent
\textbf{Correctness:} The correctness depends
on Alice computing the same $P_{\text{final}}$ as Bob.  But
this will be true because $P_{A}, P_{B}$ commute, as do
$Q_{A}, Q_{B}$, so $P_{B}\cdot P_{\text{Alice}}\cdot Q_{B}
= P_{A}\cdot P_{\text{Bob}}\cdot Q_{A}$.

\noindent
\textbf{Security:} Comparing this scheme
with Algorithm \ref{alg:keyExchange}, it is clear that
$P_{\text{Alice}}$ is the product sent from Alice to Bob,
$P_{\text{Bob}}$ the product from Bob to Alice, and
$P_{\text{final}}$ is the shared final key that is agreed upon. In
addition to knowing these exchanged values, Eve will also know the
encrypted message $m_e$.
So the security is essentially the same as that of
Algorithm \ref{alg:keyExchange}, except for an additional
assumption, that the attacker  will find it difficult to factor
$m_{e} := m\cdot P_{\text{final}}$.

\begin{remark}
As ElGamal observes, the multiplication
$m_{e} := m\cdot P_{\text{final}}$
in the encryption step could be replaced by any
invertible operation, e.g., by addition.  Indeed, setting
$m_{e} := m + P_{\text{final}}$ gives essentially the
encryption scheme proposed by Boucher et al.
\cite{boucher2010key}.  This has the advantage of being
easier to compute, and giving an encrypted $m_{e}$ of smaller size
than one that uses multiplication.  We would caution,
however, that
for such a scheme to be secure, both $m$ and $P_{\text{final}}$
should probably be dense polynomials.  If one polynomial is
dense but
the other sparse, for example, much of the structure of the dense
polynomial would be visible in $m_{e}$, and could aid an attacker.
\end{remark}

\begin{remark} 
The message encryption step (i.e. computing $m_e = m\cdot
P_{\text{final}}$) is not required to be performed inside $S$. In
fact, $P_{\text{final}}$ is the key with which Bob encrypts the
message $m$, and the algorithm that is used for this step can be
replaced by any known and well-studied private key system, such as AES
\cite{daemen2002design}. In this case, the security of our protocol
reduces to the security of 
 our Algorithm \ref{alg:keyExchange}, along with the security of the
 chosen private key system.

% The above scheme is included here only for
% completeness.  In fact, we would not recommend it as the first
% choice for encrypting $m$, as it has not been evaluated.  Such
% an evaluation could require, for example, a lengthy investigation
% of the
% structures of possible messages $m$ and their encryptions $m_{e}$,
% searching for any weaknesses this might show, and
% has not been
% attempted here. Instead, after Bob
% has computed $P_{\text{final}}$ above, we recommend that
% he use $P_{\text{final}}$ as a key for some well-studied
% private key system, such as AES \cite{daemen2002design}.
% After encrypting $m$ to
% $m_{\text{AES}}$, he then sends the pair
% $\left(m_{\text{AES}},P_{\text{Bob}}\right)$ to Alice.
% Alice then computes
% $P_{\text{final}} = P_{A}\cdot P_{\text{Bob}}\cdot Q_{A}$,
% as above, and uses it to decrypt $m_{\text{AES}}$ to $m$.

\end{remark}

\subsubsection{An ElGamal-like Digital Signature Scheme}
\label{subsubsec:ElGamal-signature}

Suppose Alice wishes to prove to Bob that a message
she is sending him did, in fact, come from her.  For his part,
Bob may want this proof, both as a guard against forgers,
and also to prevent Alice from denying at a later time
that the message had come from her.  The digital signature
scheme shown below is intended to accomplish this.
(Note that a message $m$
need not be encrypted to be signed; Alice may also use
her signature to show that she publicly approves of a
cleartext message.)  As in the encryption scheme in section
\ref{subsubsec:ElGamal-encryption}, this requires Alice
to precompute and publish information at some location,
e.g., a secure webpage.  This must be securely associated with her,
both to prevent forgers from altering the information,
and to prevent her from trying to repudiate any
messages bearing her digital signature. Again, as in section
\ref{subsubsec:ElGamal-encryption}, our proposed signature scheme is
only intended to provide a basic level of security (i.e. Eve cannot
forge Alice's signature on an arbitrary message of Eve's
choice). Further enhancements are needed for stronger notions of security.

\noindent
\textbf{Preparation:} Alice chooses $a_1,L,a_2\in S$, pairwise
non-commutative.
Alice publishes $L$ and
$P_{\text{Alice}}:=a_1\cdot L \cdot a_2$.
The tuple $(a_1, a_2)$ is kept secret.

\noindent
\textbf{Signature Creation:} Let $m \in S$ be a message
which Alice wishes to sign.  ($m$ may be either encrypted
or unencrypted.)

Alice chooses two new random polynomials
$k_1, k_2 \in S$. (Alice checks that $L$, $k_{1}$, $k_{2}$
are pairwise non-commutative.)  Alice computes
$\gamma := k_{1}\cdot L\cdot k_{2}$,
$\epsilon_{1}:=k_{1}\cdot L\cdot a_{2}$, and
$\epsilon_{2}:=a_{1}\cdot L\cdot k_{2}$.
The tuple $(k_{1}, k_{2})$ is kept secret.

Now Alice chooses $q_1\in S$ and computes
$r_1 := m - \gamma a_{1} - q_{1}k_{1} \in S$, so that:
\begin{equation}  \label{eq:m_div_1}
  m - \gamma a_{1} = q_{1}k_{1} + r_{1}.
\end{equation}
Analogously, Alice also chooses $q_2 \in S$ and computes
$r_2 := m - a_{2}\gamma - k_{2}q_{2} \in S$, so that:
\begin{equation}  \label{eq:m_div_2}
  m - a_{2}\gamma = k_{2}q_{2} + r_{2}.
\end{equation}% (As $S$ is not a Euclidean domain, the pairs
% $\left(q_{1},r_{1}\right)$ and $\left(q_{2},r_{2}\right)$
% are not necessarily unique, but will depend on some ordering
% being assigned to the variables.  For the purposes of the
% signature, however, it does not matter which ordering is used,
% as long as we find some $k_{1}$, $r_{1}$, $k_{2}$, $r_{2}$
% satisfying (\ref{eq:m_div_1}) and (\ref{eq:m_div_2})).

\noindent
Alice sends Bob the signed message $m$ as the 8-tuple
$\left(m,\gamma,q_{1},r_{1},q_{2},r_{2},\epsilon_{1},
\epsilon_{2}\right)$.  Note that $k_{1}$, $k_{2}$,
$a_{1}$, $a_{2}$ are all kept
secret by Alice.

\noindent
\textbf{Signature Verification:}
Given $\left(m,\gamma,q_{1},r_{1},q_{2},r_{2},\epsilon_{1},
\epsilon_{2}\right)$, (and the public $L$
and $P_{\text{Alice}})$, Bob computes:
\begin{equation}
  \begin{split}
    \mathrm{sig}_{\text{left}} &:= \left(m-r_{1}\right)\cdot L
                          \cdot \left(m-r_{2}\right),\\
    \mathrm{sig}_{\text{right}} &:= q_{1}\gamma q_{2}
                           + q_{1}\epsilon_{1}\gamma
                           + \gamma \epsilon_{2} q_{2}
                           + \gamma P_{\text{Alice}} \gamma.
  \end{split}
\end{equation}If $\mathrm{sig}_{\text{left}}=\mathrm{sig}_{\text{right}}$,
then the signature is accepted as valid.  Otherwise, the
signature is rejected.

\noindent
\textbf{Correctness:} We have
\begin{equation}  \nonumber
  \begin{split}
    \mathrm{sig}_{\text{left}} &= \left(m-r_{1}\right)\cdot L
                          \cdot \left(m-r_{2}\right) \\
                      &= \left(q_{1}k_{1} + \gamma a_{1}\right)
                         \cdot L
                         \cdot \left(k_{2}q_{2} + a_{2}\gamma\right)
                           \text{,\ \ using \ (\ref{eq:m_div_1}),
                                 (\ref{eq:m_div_2})} \\
                      &= q_{1}k_{1} L k_{2}q_{2}
                           + q_{1}k_{1} L a_{2}\gamma
                           + \gamma a_{1} L k_{2}q_{2}
                           + \gamma a_{1} L a_{2} \gamma \\
                      &= q_{1}\gamma q_{2}
                           + q_{1}\epsilon_{1}\gamma
                           + \gamma \epsilon_{2} q_{2}
                           + \gamma P_{\text{Alice}} \gamma \\
                      &= \mathrm{sig}_{\text{right}},
  \end{split}
\end{equation}using $k_{1} L k_{2} = \gamma$,
\, $k_{1} L a_{2} = \epsilon_{1}$,
\, $a_{1} L k_{2} = \epsilon_{2}$,
\, and \, $a_{1} L a_{2} = P_{\text{Alice}}$.

\noindent
\textbf{Security:}  To forge a signature for a given message $m$,
Eve must find
values of $\gamma$, $q_{1}$, $r_{1}$, $q_{2}$, $r_{2}$, $\epsilon_{1}$,
$\epsilon_{2}$ which result in 
$\mathrm{sig}_{\text{left}}= \mathrm{sig}_{\text{right}}$.  We cannot see any
easy way to do this.  She can, of course, choose her own values
of $k_{1}$, $k_{2}$, and compute a corresponding
$\gamma=k_{1}Lk_{2}$.  This is, indeed, a plausible value for
$\gamma$,
since Alice could have chosen these values for $k_{1}$, $k_{2}$
herself.  But Eve does not have the values of $a_{1}$ and
$a_{2}$.  If she guesses them, as say $\tilde{a}_{1}$,
$\tilde{a}_{2}$, then eventually she will fail to form
$P_{\text{Alice}}$ correctly when the product
$\tilde{a}_{1}L\tilde{a}_{2}$ is formed in the correctness proof above,
and will not have
$\mathrm{sig}_{\text{left}}=\mathrm{sig}_{\text{right}}$.

Furthermore, given a legitimate signature from Alice,
it appears doubtful that an attacker can recover
$a_{1}$ and $a_{2}$.  For example, from (\ref{eq:m_div_1}),
we have $m-r_{1}=q_{1}k_{1} + \gamma a_{1}$. Here, the left
side is known, as are $q_{1}$ and $\gamma$ on the right.
But we know no way of determining $k_{1}$ and $a_{1}$,
other than computing the syzygy-module $M \in S^3$ of $q_1$, $\gamma$, and
$m-r_1$ and considering the subset $\{\ell = [\ell_1, \ell_2, \ell_3] \in M \mid \ell_3 =-1\}
\subseteq M$. As this is in general an infinite set, this does not
yield a practical way of
recovering $k_1$ and $a_1$.
% setting up an ansatz, and trying to solve the
% resulting nonlinear system.

It will be noted, however, that computing a signature
may result in very large expressions.  Thus,
$\left(m-r_{1}\right)\cdot L\cdot \left(m-r_{2}\right)$
is likely to be much larger than the message $m$ itself.
However, in practice, one would only create a signature on a certain
hash-sum of the message $m$, which has fixed length for all possible messages.
% It would be desirable, therefore, to have a signature
% scheme that uses much less space, ideally not much more
% than $m$ does itself.  Nonetheless, the above scheme
% at least shows that a signature scheme using
% non-commutative polynomial rings is possible.

\subsection{A Zero Knowledge Proof Protocol Using Multivariate
Ore Polynomials}

As usual, let $S:=R[\partial_1;\sigma_1,\delta_1]
[\partial_2;\sigma_2,\delta_2]
\ldots [\partial_n;\sigma_n,\delta_n]$, for some fixed
$n \ge 2$.  Let $L \in S$.  Suppose Alice knows a factorization
(possibly partial) of $L$ into two nontrivial factors
$\ell_{1}$ and $\ell_{2}$.  That is, $L= \ell_{1}\cdot \ell_{2}$,
where $L, \ell_{1}, \ell_{2} \in S$.  (By nontrivial, we mean
that for at least one $i$, where $1 \le i \le n$,
we have $\deg_{\partial_{i}}(\ell_{1}) \ge 1$, and similarly
for $\ell_{2}$.  Furthermore, the factorization may only be partial,
i.e., $\ell_{1}$ or $\ell_{2}$
may themselves factor over $S$.  But we will not be concerned
here with this possibility.)  Note that as we are assuming that
factorization of multivariate Ore polynomials is, in general,
computationally infeasible, Alice may well have accomplished her
feat by first choosing $\ell_{1}$ and $\ell_{2}$, then creating
a suitable $L$ by computing $L=\ell_{1}\cdot \ell_{2}$, and then
finally publishing $L$.
But the particular method she used to produce an $L$ and its
factors is not important to the protocol.
 
Bob also knows $L$, but does not know $\ell_{1}$ or $\ell_{2}$.
Alice wishes to convince Bob, beyond any reasonable doubt,
that she knows a factorization of $L$.  But she does not wish
to tell Bob her factors $\ell_{1}$ or $\ell_{2}$, or to give Bob 
enough information that would allow him to compute the
factors within any feasible time.  We have, therefore,
a situation which calls for a zero knowledge proof protocol.
 
Alice can do this using the following protocol, which is
repeated as many times as Bob desires, until he is convinced
that Alice does, indeed, possess a factorization of $L$:

\noindent
\textbf{Step 1:}  Alice chooses two polynomials
$p_{1}, p_{2} \in S$, and forms the product
$\pi := p_{1}\cdot L \cdot p_{2}$.
She sends $\pi$ to Bob.  She also tells Bob the degrees
of $p_{1}$ and $p_{2}$ in each of the $\partial_{i}$.
That is, for each $i$, where $1 \le i \le n$, she sends
Bob $\deg_{\partial_{i}}(p_{1})$ and $\deg_{\partial_{i}}(p_{2})$.
Apart from this information, however, Alice keeps $p_{1}$ and
$p_{2}$ private, unless Bob specifically asks for them in Step 2.
 
Note that $\pi = p_{1}\cdot L \cdot p_{2}
= p_{1} \cdot \ell_{1}\ell_{2} \cdot p_{2}
= \pi_{1} \cdot \pi_{2}$, where we define
$\pi_{1} := p_{1} \ell_{1}$, and $\pi_{2} := \ell_{2} p_{2}$.
Thus, two different partial factorizations of $\pi$ are
$\pi = p_{1}\cdot L \cdot p_{2}$ and $\pi = \pi_{1} \cdot \pi_{2}$.
Essentially, in Step 2., Bob may ask Alice for one, and only one,
of these two factorizations.

\noindent
\textbf{Step 2:}  Having received $\pi$ and the degree information
for $p_{1}$ and $p_{2}$ from Alice, Bob asks Alice for exactly one
of the following: Either (i) $p_{1}$ and $p_{2}$, or (ii) $\pi_{1}$
and $\pi_{2}$.
 
\noindent
\textbf{Step 2a:}  If Bob asked for $p_{1}$ and $p_{2}$, he
checks that $\pi = p_{1}\cdot L \cdot p_{2}$.  He also checks
that $p_{1}$ and $p_{2}$ satisfy the degree bounds sent earlier
by Alice in Step 1.

\noindent
\textbf{Step 2b:}  If Bob asked for $\pi_{1}$ and $\pi_{2}$, he     
checks that $\pi = \pi_{1} \cdot \pi_{2}$.  He also checks that
$\pi_{1}$ and $\pi_{2}$ satisfy the following degree conditions:
For each $i$, where $1 \le i \le n$, he requires that
$\deg_{\partial_{i}}(\pi_{1}) \ge \deg_{\partial_{i}}(p_{1})$.
Furthermore, for at least
one $j$, $1 \le j \le n$, the inequality must be strict, i.e.,
$\deg_{\partial_{j}}(\pi_{1}) > \deg_{\partial_{j}}(p_{1})$.
Likewise, $\pi_{2}$
must satisfy $\deg_{\partial_{i}}(\pi_{2}) \ge \deg_{\partial_{i}}(p_{2})$
for all $i$,
where $1 \le i \le n$, and for at least one $k$, where
$1 \le k \le n$, must satisfy
$\deg_{\partial_{k}}(\pi_{2}) > \deg_{\partial_{k}}(p_{2})$.

\noindent
If any of these checks fail, then the
protocol terminates, and Bob rejects Alice's claim to know a
nontrivial factorization of $L$.  If the checks hold true,
then Alice has passed this cycle of the protocol.
 
\noindent
Steps 1. through 2b. constitute one cycle of the protocol.
For each cycle, Alice chooses a new pair of polynomials
$p_{1}$, $p_{2}$ in Step 1., never using the same polynomial
more than once.  The cycle is repeated until either one of
the checks fail, or until Bob is convinced, beyond a reasonable
doubt, that Alice does indeed posses a nontrivial factorization
of $L$.

\noindent
\textbf{Discussion and Security of the Protocol:}
For each cycle of the protocol, Bob can randomly choose
whether to ask Alice for $p_{1}$, $p_{2}$, or for
$\pi_{1}$, $\pi_{2}$.  Suppose this is repeated many times,
with the answers always satisfying the checks.  Then Bob
should eventually be convinced, beyond a reasonable
doubt, that for any $\pi$ offered by Alice in Step 1.,
she is always able to factor $\pi$ both as $p_{1}Lp_{2}$
and as $\pi_{1}\pi_{2}$, with the degree conditions
also satisfied.
 
The two questions we must address are: (i) should Bob
believe that Alice knows some nontrivial factorization
$\ell_{1}\ell_{2}$ of $L$?; and (ii) do Alice's answers
to Bob's queries give Bob a practical method to determine
a nontrivial factorization of $L$?
 
\noindent
\textbf{(i) Alice can factor $L$:}

First, consider the requirement that Alice must give degree
conditions on $p_{1}$, $p_{2}$ in Step 1., before Bob says
whether he wishes to know $p_{1}$, $p_{2}$, or $\pi_{1}$,
$\pi_{2}$ in Step 2.  If we did not impose this condition
on Alice, she could trick Bob into believing that she had
a factorization of $L$ as follows.  First, instead of a $p_{1}$,
she chooses $p_{1}', p_{1}'', p_{2} \in S$, and computes
$\pi := p_{1}' \cdot p_{1}'' \cdot L \cdot p_{2}$, sending
only $\pi$ to Bob.  Now, if Bob asks for $p_{1}$, $p_{2}$,
she sends him $p_{1} := p_{1}'p_{1}''$, and $p_{2}$.  He
then computes $p_{1} \cdot L \cdot p_{2}$, and finds this
product to equal $\pi$, as expected.  On the other hand,
if he asks for $\pi_{1}$, $\pi_{2}$, Alice sends him
$\pi_{1} := p_{1}'$, $\pi_{2} := p_{1}''Lp_{2}$.  Again,
Bob will find that $\pi_{1} \cdot \pi_{2}$ yields $\pi$,
as expected.  So Alice appears to have passed the test,
even though she needed no knowledge of any factorization
of $L$.

Forcing Alice to give the degree conditions in Step 1.,
before Bob announces his choice in Step 2., prevents
Alice from using this strategy to deceive Bob.
Furthermore, the degree conditions on $\pi_{1}$ and
$\pi_{2}$, if Bob should ask for these two polynomials,
prevents Alice from simply setting either
$\pi_{1} := p_{1}$, $\pi_{2} := Lp_{2}$, or
$\pi_{1} := p_{1}L$, $\pi_{2} := p_{2}$, neither
of which would have required any knowledge of the factors
of $L$.  Informally, the degree conditions force both
$\pi_{1}$ and $\pi_{2}$ to partially "overlap" $L$
in the product $\pi$.
 
Now, the above might appear to imply that $\pi_{1}$
and $\pi_{2}$ must have the forms
$\pi_{1} = p_{1}\ell_{1}$, $\pi_{2} = \ell_{2}p_{2}$.
However, factorization in $S$ is unique only up to
similarity
(cf. Section \ref{subsec:basic-notations-and-definitions},
Definition \ref{definition:similarity}).  Consequently,
it is possible that Alice has found $p_{1}$, $p_{2}$,
$\pi_{1}$, $\pi_{2}$, such that
$p_{1}Lp_{2}=\pi_{1}\pi_{2}$ (which she sets equal to $\pi$),
and yet $p_{1}$ is not a left divisor of $\pi_{1}$,
nor is $p_{2}$ a right divisor of $\pi_{2}$.  With such
a choice of polynomials, Alice would be able to satisfy
the demands of
the protocol, without necessarily having any knowledge
of the factors of $L$.  However, finding such a set of
polynomials is, as far as we know, computationally
impractical with currently known methods.  The only
reasonable conclusion for Bob, therefore, is that
$\pi_{1} = p_{1}\ell_{1}$, $\pi_{2} = \ell_{2}p_{2}$.

However, if Alice knows both $\pi_{1}$ and $p_{1}$,
she can simply perform exact division of $\pi_{1}$ on
the left by $p_{1}$ to obtain $\ell_{1}$.  Similarly,
(exactly) dividing $\pi_{2}$ on the right by $p_{2}$
will yield $\ell_{2}$.  Hence Bob must conclude that
Alice does, indeed, know some nontrivial factorization
$\ell_{1}\ell_{2}$ of $L$.
 
\noindent
\textbf{(ii) Bob cannot factor $L$:}
 
Again, let us first consider the effect that knowing
the degrees of $p_{1}$ and $p_{2}$ will have on Bob's
attempts to find a factorization of $L$.  Certainly,
this knowledge makes it easier for him to set up an
ansatz, e.g., of the form
$\pi = p_{1}\ell_{1}\ell_{2}p_{2}$.  However, even
if Bob did not know these degree conditions, the
number of possibilities he would have to consider
would increase only by a factor that is polynomial
in $n$ and the maximum degree in any $\partial_{i}$
of $\pi$.  That is, if Bob could find some factors
of $L$ in polynomial time by some algorithm which
makes use of the given degree conditions, he could
also find the factors in polynomial time without
knowing these conditions.  Thus, these degree conditions
do not, in themselves, compromise the security of the
protocol.

Now, Bob obtains two sorts of information from Alice.
The first is $p_{1}$ and $p_{2}$ such that
$\pi = p_{1}Lp_{2}$.  Clearly, this does not help
Bob to factor $L$, since he could just as well have
chosen his own $p_{1}$ and $p_{2}$, computed the
product $p_{1}Lp_{2}$ to form his own $\pi$, and
done so as many times as he wished.  Such queries,
therefore, do not lead to any exploitable weakness.
 
The second type of query gives Bob $\pi_{1}$ and
$\pi_{2}$ such that $\pi = \pi_{1}\pi_{2}$.  Let
us first consider $\pi_{1}$ (the situation for
$\pi_{2}$ is similar): Bob knows that
$\pi_{1} = p_{1}\ell_{1}$, though of course
$p_{1}$ and $\ell_{1}$ are unknown to him.
Bob also knows $L$, though again its factors
$\ell_{1}$, $\ell_{2}$ are also unknown to him.
Hence one type of attack would be to use the
fact that $\ell_{1}$ is a right divisor of
$\pi_{1}$, and also a left divisor of $L$.
It is interesting to note, however, that even
in (non-commutative) Euclidean domains, there
is currently no known practical algorithm to
find such a left and right simultaneous
divisor.  Our domain $S$ is not even Euclidean,
so the situation is even worse for an attacker.
The only general, though impractical, approach
involves forming an ansatz,
and solving the resulting quadratic           
system of multivariate polynomial equations.
So this attack fails.
 
Another approach Bob might try is to ask
many queries of this sort, to generate
a growing set of polynomials
$\pi = p_{1}\ell_{1}$,
$\pi' = p_{1}'\ell_{1}$,
$\pi'' = p_{1}''\ell_{1}$, et cetera.
Now, $\ell_{1}$ is a common right divisor
of all these polynomials, and as the set grows
in size, it will very likely be the gcrd
of the set.  Hence, if $S$ were a Euclidean
domain, Bob could use the Euclidean
algorithm to find $\ell_{1}$ with high
probability.  However, $S$
is not a Euclidean domain.
The only hope for a potential attacker could be -- similar to the
attack described in section \ref{sbsbsctn:duboiskammerer} --  a computation of a
left Gr\"obner basis of the ideal generated by the $\pi$s, if the notion of
a left Gr\"obner basis exists for the chosen ring $S$ (which does not,
if $S$ is chosen non-Noetherian). Then, it is possible
that the basis consists of only one element, namely $\ell_1$. We made
experiments in the polynomial second Weyl algebra, choosing very small degrees
(total degree less or equal to five) for $\ell_1$. With a
sufficiently large number of these $\pi$s, the computed Gr\"obner
basis had indeed the desired structure. However,
as the computation of a Gr\"obner basis is exponential space hard
\cite{mayr1982complexity}, this is not a feasible approach in
general. Even with slightly larger choices for these degrees, the
Gr\"obner basis computation in case of the second Weyl algebra fails to terminate
within a reasonable time.
Other than this, we are not aware of any known practical method to find
such a common divisor for $S$.

Again, a general, though impractical, approach
involves forming an ansatz,
and solving the resulting quadratic
system of multivariate polynomial equations.

Similar remarks apply to Bob's attempts to
determine $\ell_{2}$.
 
We conclude, therefore, that it will be
impractical for Bob to determine a
nontrivial factorization of $L$ by using
Alice's answers to his queries.

\section{Conclusion}
\label{sctn:conclusion}

The key exchange primitive as presented in \cite{boucher2010key} has been
altered to be immune against the attack presented in
\cite{dubois2011cryptanalysis}, and extended. The new version presented
in this paper continues to have the
positive properties discussed by Boucher et al. in their
paper. The security of our proposal is related to the hardness
of factoring in non-commutative rings and the non-uniqueness of the
factorization. A class of insecure key choices that would reduce the
problem to commutative factorization was outlined.
Moreover, we provide the freedom to choose rings that are
not Noetherian, where a general factorization
algorithm might not even exist.

An implementation for a specific ring 
is provided and we look forward to feedback from
users. Furthermore, based on this implementation,
we have published some challenge problems as described
in section \ref{sbsctn:challenges}. We encourage the reader to attack
our proposed schemes via these challenges.

We also mention here that related protocols can be developed 
using this primitive.  We presented four such enhancements in
section \ref{sec:Enhancements}:
a three-pass protocol, an ElGamal-like encryption scheme,
and an ElGamal-like digital signature scheme, and a zero-knowledge
proof protocol.

Of course the security and practicality of our protocols
need to be examined further for particular choices of the
non-commutative ring $S$, which is described in as general a
way as possible in this article.  All of them can be broken if an algorithm to
find a specific factorization in a feasible amount of time is
available. However, researchers have been interested in factoring Ore polynomials
for decades, and this problem is in general perceived as very
difficult. As Ore polynomial rings are in general abstractions of operator
algebras, any success of breaking our protocols would lead to a
further understanding of the underlying operators. Thus any
successful attempt to break our protocols, even if it is just for one
special choice of $S$, would benefit several
scientific communities.

Moreover, we note that some of
our proposed schemes bear a strong resemblance to others that have been
known, studied, and withstood general attacks for years, or even
decades.  For example, our digital signature scheme can be viewed
as an analogue of the ElGamal scheme, adapted for non-commutative
rings.  But to the best of our knowledge, no proof of security,
showing that breaking the ElGamal signature scheme is equivalent to
solving the discrete logarithm problem, has yet been found.
Instead, we have the accumulated experience of many cryptographers,
who have so far found no general method of attack that does not
involve computing a discrete logarithm. This appears to be the
basis for accepting the scheme as secure. Now, given the similarity
of our proposal to that of ElGamal, one can plausibly suggest that
ours will also be secure, unless -- again -- some practical method of
factorization is found.
 
Interesting questions for future research are: For which choices of
a ring $S$ of
type (\ref{eq:ourRings}) can one construct an effective attack for the
proposed schemes (possibly using quantum computers)? Note, that the
rings we chose for our examples and for our implementations are among
the simplest ones (Noetherian, bivariate, over finite fields)
that appear to be immune to known attacks. And furthermore, can one improve
the computation of the arithmetics in those rings using
a quantum computer?

We also hope for better implementations in the future for arithmetics
in Ore polynomials, since the existing ones in commodity computer algebra
systems appear to be slow on large examples. This fact forced us to
write our own experimental implementation to evaluate the feasibility
of our proposals.
\begin{comment}
 - In general, we see the future of this in using model algebras like
    the Weyl algebras and discussing (in-)secure keys in them and
    possible attacks.
\end{comment}

\section*{Acknowledgements}
We thank Konstantin Ziegler for
his valuable remarks, comments and suggestions on this
paper.

\noindent
We are grateful to Alfred Menezes for the fruitful discussion we had with him.
\bibliography{burgerHeinle}

\newcommand{\etalchar}[1]{$^{#1}$}
\newcommand{\Hoeven}{\relax}
\begin{thebibliography}{MGH{\etalchar{+}}08}

\bibitem[Art47]{artin1947theory}
Emil Artin.
\newblock Theory of braids.
\newblock {\em Annals of Mathematics}, pages 101--126, 1947.

\bibitem[BC23]{burchnall1923commutative}
J.L. Burchnall and T.W. Chaundy.
\newblock Commutative ordinary differential operators.
\newblock {\em Proceedings of the London Mathematical Society}, 2(1):420--440,
  1923.

\bibitem[BGG{\etalchar{+}}10]{boucher2010key}
Delphine Boucher, Philippe Gaborit, Willi Geiselmann, Olivier Ruatta, and Felix
  Ulmer.
\newblock Key exchange and encryption schemes based on non-commutative skew
  polynomials.
\newblock In {\em Post-Quantum Cryptography}, pages 126--141. Springer, 2010.

\bibitem[BGTV03]{Bueso:2003}
J.~Bueso, J.~G\'omez-Torrecillas, and A.~Verschoren.
\newblock {\em {Algorithmic methods in non-commutative algebra. Applications to
  quantum groups.}}
\newblock {Dordrecht: Kluwer Academic Publishers}, 2003.

\bibitem[BK05]{beals2005constructively}
R.~Beals and Elena~A. Kartashova.
\newblock Constructively factoring linear partial differential operators in two
  variables.
\newblock {\em Theor. Math. Phys.}, 145(2):1511--1524, 2005.

\bibitem[BLP08]{bernstein2008attacking}
Daniel~J. Bernstein, Tanja Lange, and Christiane Peters.
\newblock Attacking and defending the mceliece cryptosystem.
\newblock In {\em Post-Quantum Cryptography}, pages 31--46. Springer, 2008.

\bibitem[Bro01]{brown2001quest}
Julian Brown.
\newblock {\em Quest for the quantum computer}.
\newblock Simon and Schuster, 2001.

\bibitem[Buc97]{Buchberger:1997}
B.~Buchberger.
\newblock {Introduction to Groebner bases.}
\newblock {Berlin: Springer}, 1997.

\bibitem[CDW07]{cao2007new}
Zhenfu Cao, Xiaolei Dong, and Licheng Wang.
\newblock New public key cryptosystems using polynomials over non-commutative
  rings.
\newblock {\em IACR Cryptology ePrint Archive}, 2007:9, 2007.

\bibitem[CJ03]{cheon2003polynomial}
Jung~Hee Cheon and Byungheup Jun.
\newblock A polynomial time algorithm for the braid diffie-hellman conjugacy
  problem.
\newblock In {\em Advances in Cryptology-CRYPTO 2003}, pages 212--225.
  Springer, 2003.

\bibitem[CKPS00]{courtois2000efficient}
Nicolas Courtois, Alexander Klimov, Jacques Patarin, and Adi Shamir.
\newblock Efficient algorithms for solving overdefined systems of multivariate
  polynomial equations.
\newblock In {\em Advances in Cryptology---EUROCRYPT 2000}, pages 392--407.
  Springer, 2000.

\bibitem[CNT12]{climent2012key}
Joan-Josep Climent, Pedro~R Navarro, and Leandro Tortosa.
\newblock Key exchange protocols over noncommutative rings. the case of
  $\mathrm{End}(\mathbb{Z}_p\times \mathbb{Z}_{p^2})$.
\newblock {\em International Journal of Computer Mathematics},
  89(13-14):1753--1763, 2012.

\bibitem[CVHL10]{cha2010solving}
Yongjae Cha, Mark Van~Hoeij, and Giles Levy.
\newblock Solving recurrence relations using local invariants.
\newblock In {\em Proceedings of the 2010 International Symposium on Symbolic
  and Algebraic Computation}, pages 303--309. ACM, 2010.

\bibitem[DGPS12]{Singular:2012}
W.~Decker, G.-M. Greuel, G.~Pfister, and H.~Sch{\"o}nemann.
\newblock {\sc Singular} {3-1-6} --- {A} computer algebra system for polynomial
  computations.
\newblock 2012.
\newblock \url{http://www.singular.uni-kl.de}.

\bibitem[DH76]{diffie1976new}
Whitfield Diffie and Martin~E Hellman.
\newblock New directions in cryptography.
\newblock {\em Information Theory, IEEE Transactions on}, 22(6):644--654, 1976.

\bibitem[DK11]{dubois2011cryptanalysis}
Vivien Dubois and Jean-Gabriel Kammerer.
\newblock Cryptanalysis of cryptosystems based on non-commutative skew
  polynomials.
\newblock In {\em Public Key Cryptography--PKC 2011}, pages 459--472. Springer,
  2011.

\bibitem[DR02]{daemen2002design}
Joan Daemen and Vincent Rijmen.
\newblock {\em The design of Rijndael: AES-the advanced encryption standard}.
\newblock Springer, 2002.

\bibitem[ElG85]{elgamal1985public}
Taher ElGamal.
\newblock A public key cryptosystem and a signature scheme based on discrete
  logarithms.
\newblock In {\em Advances in Cryptology}, pages 10--18. Springer, 1985.

\bibitem[Gar86]{garling1986course}
David~JH Garling.
\newblock {\em A course in Galois theory}.
\newblock Cambridge University Press, 1986.

\bibitem[GHL14]{giesbrecht2014factoring}
M.~Giesbrecht, A.~Heinle, and V.~Levandovskyy.
\newblock Factoring linear differential operators in $n$ variables.
\newblock In {\em Proceedings of the 39th International Symposium on Symbolic
  and Algebraic Computation}, ISSAC '14, pages 194--201, New York, NY, USA,
  2014. ACM.

\bibitem[Gie98]{giesbrecht1998factoring}
Mark Giesbrecht.
\newblock Factoring in skew-polynomial rings over finite fields.
\newblock {\em Journal of Symbolic Computation}, 26(4):463--486, 1998.

\bibitem[GS04]{GrigorievSchwartz:2004}
D.~Grigoriev and F.~Schwarz.
\newblock {Factoring and solving linear partial differential equations.}
\newblock {\em Computing}, 73(2):179--197, 2004.

\bibitem[GZ03]{giesbrecht2003factoring}
Mark Giesbrecht and Yang Zhang.
\newblock Factoring and decomposing ore polynomials over fq(t).
\newblock In {\em Proceedings of the 2003 International Symposium on Symbolic
  and Algebraic Computation}, ISSAC '03, pages 127--134, New York, NY, USA,
  2003. ACM.

\bibitem[HL13]{heinle2013factorization}
Albert Heinle and Viktor Levandovskyy.
\newblock Factorization of $\mathbb{Z}$-homogeneous polynomials in the first
  $(q)-${W}eyl algebra.
\newblock {\em arXiv preprint arXiv:1302.5674}, 2013.

\bibitem[Jac43]{nathan1943theory}
Nathan Jacobson.
\newblock {\em The theory of rings}.
\newblock Number~2. American Mathematical Soc., 1943.

\bibitem[KLC{\etalchar{+}}00]{ko2000new}
Ki~Hyoung Ko, Sang~Jin Lee, Jung~Hee Cheon, Jae~Woo Han, Ju-sung Kang, and
  Choonsik Park.
\newblock New public-key cryptosystem using braid groups.
\newblock In {\em Advances in cryptology---CRYPTO 2000}, pages 166--183.
  Springer, 2000.

\bibitem[Kra02]{krammer2002braid}
Daan Krammer.
\newblock Braid groups are linear.
\newblock {\em Annals of Mathematics}, pages 131--156, 2002.

\bibitem[KS99]{kipnis1999cryptanalysis}
Aviad Kipnis and Adi Shamir.
\newblock Cryptanalysis of the hfe public key cryptosystem by relinearization.
\newblock In {\em Advances in cryptology---CRYPTO'99}, pages 19--30. Springer,
  1999.

\bibitem[Lan02]{landau1902satz}
Edmund Landau.
\newblock Ein {S}atz {\"u}ber die {Z}erlegung homogener linearer
  {D}ifferentialausdr{\"u}cke in irreducible {F}actoren.
\newblock {\em Journal f{\"u}r die reine und angewandte Mathematik},
  124:115--120, 1902.

\bibitem[Loe03]{Loewy:1903}
A.~Loewy.
\newblock {\"Uber reduzible lineare homogene Differentialgleichungen.}
\newblock {\em Math. Ann.}, 56:549--584, 1903.

\bibitem[Loe06]{Loewy:1906}
A.~Loewy.
\newblock {\"Uber vollst\"andig reduzible lineare homogene
  Differentialgleichungen.}
\newblock {\em Math. Ann.}, 62:89--117, 1906.

\bibitem[MA94]{Melenk:1994}
H.~Melenk and J.~Apel.
\newblock {\em REDUCE package NCPOLY: Computation in non-commutative polynomial
  ideals.}
\newblock Konrad-Zuse-Zentrum Berlin (ZIB), 1994.

\bibitem[Mau94]{maurer1994towards}
Ueli~M Maurer.
\newblock Towards the equivalence of breaking the diffie-hellman protocol and
  computing discrete logarithms.
\newblock In {\em Advances in cryptology---CRYPTO'94}, pages 271--281.
  Springer, 1994.

\bibitem[McE78]{mceliece1978public}
Robert~J McEliece.
\newblock A public-key cryptosystem based on algebraic coding theory.
\newblock {\em DSN progress report}, 42(44):114--116, 1978.

\bibitem[MGH{\etalchar{+}}08]{Maple}
M.~B. Monagan, K.~O. Geddes, K.~M. Heal, G.~Labahn, S.~M. Vorkoetter,
  J.~McCarron, and P.~DeMarco.
\newblock {\em Maple Introductory Programming Guide}.
\newblock Maplesoft, 2008.

\bibitem[MM82]{mayr1982complexity}
Ernst~W Mayr and Albert~R Meyer.
\newblock The complexity of the word problems for commutative semigroups and
  polynomial ideals.
\newblock {\em Advances in mathematics}, 46(3):305--329, 1982.

\bibitem[MR01]{mcconnell2001noncommutative}
John~C McConnell and James~Christopher Robson.
\newblock {\em Noncommutative noetherian rings}, volume~30.
\newblock American Mathematical Soc., 2001.

\bibitem[Ore33]{ore1933theory}
Oystein Ore.
\newblock Theory of non-commutative polynomials.
\newblock {\em Annals of mathematics}, 34:480--508, 1933.

\bibitem[Sch09]{Schwarz:2009}
F.~Schwarz.
\newblock Alltypes in the web.
\newblock {\em ACM Commun. Comput. Algebra}, 42(3):185--187, February 2009.

\bibitem[She07]{parameters:shemyakova:2007}
Ekaterina Shemyakova.
\newblock Parametric factorizations of second-, third- and fourth-order linear
  partial differential operators with a completely factorable symbol on the
  plane.
\newblock {\em Mathematics in Computer Science}, 1(2):225--237, 2007.

\bibitem[She09]{shemyakova:multfacts1:2009}
Ekaterina Shemyakova.
\newblock Multiple factorizations of bivariate linear partial differential
  operators.
\newblock In Vladimir Gerdt, Ernst Mayr, and Evgenii Vorozhtsov, editors, {\em
  Computer Algebra in Scientific Computing}, volume 5743 of {\em Lecture Notes
  in Computer Science}, pages 299--309. Springer Berlin / Heidelberg, 2009.
\newblock 10.1007/978-3-642-04103-7\_26.

\bibitem[She10]{2010:shemyakova:refinement}
Ekaterina Shemyakova.
\newblock Refinement of two-factor factorizations of a linear partial
  differential operator of arbitrary order and dimension.
\newblock {\em Mathematics in Computer Science}, 4:223--230, 2010.
\newblock 10.1007/s11786-010-0052-3.

\bibitem[Tsa94]{Tsarev:1994}
S.P. Tsarev.
\newblock {Problems that appear during factorization of ordinary linear
  differential operators.}
\newblock {\em Program. Comput. Softw.}, 20(1):27--29, 1994.

\bibitem[Tsa96]{Tsarev:1996}
S.P. Tsarev.
\newblock {An algorithm for complete enumeration of all factorizations of a
  linear ordinary differential operator}.
\newblock In {\em Proc. ISSAC 1996}, pages 226--231. {New York, NY: ACM Press},
  1996.

\bibitem[vH96]{Hoeij:1996}
M.~van Hoeij.
\newblock {\em Factorization of linear differential operators}.
\newblock Nijmegen, 1996.

\bibitem[vH97a]{Hoeij:1997}
M.~van Hoeij.
\newblock {Factorization of differential operators with rational functions
  coefficients.}
\newblock {\em J. Symb. Comput.}, 24(5):537--561, 1997.

\bibitem[vH97b]{van1997formal}
M.~van Hoeij.
\newblock Formal solutions and factorization of differential operators with
  power series coefficients.
\newblock {\em J. Symb. Comput.}, 24(1):1--30, 1997.

\bibitem[vHY10]{Hoeij:2010}
M.~van Hoeij and Q.~Yuan.
\newblock {Finding all Bessel type solutions for linear differential equations
  with rational function coefficients}.
\newblock In {\em Proc. ISSAC 2010}, pages 37--44, 2010.

\end{thebibliography}

\section*{Appendix}
\appendix
\setcounter{remark}{0}
    \renewcommand{\theremark}{\Alph{section}\arabic{remark}}

\section{Code-Examples}
\subsection{Factorization of Chebyshev Differential Operators
using \textsc{MAPLE}}
\label{app:Chebyshev}

A few simple \textsc{MAPLE} commands to set up and find
right hand factors of
the Chebyshev differential operator (\ref{eq:chebyshev_L})
of section \ref{subsec:potential-as-a-post-quantum-ryptosystem},
with parameter $n$.
In the code, $\textbf{D} = \partial_{1}$, $\textbf{x} = x_{1}$.

%{\small
\begin{verbatim}
> with(DEtools):

> _Envdiffopdomain:=[D,x]:

> n := 11;                
                                  n := 11

> L := (1-x^2)*D^2 - x*D + n^2;
                                  2   2
                       L := (1 - x ) D  - x D + 121


> v := DFactorLCLM(L);
                  2        4         6         8         10
          1 - 60 x  + 560 x  - 1792 x  + 2304 x  - 1024 x
v := [D - -------------------------------------------------,
                   3        5        7        9   1024  11
           x - 20 x  + 112 x  - 256 x  + 256 x  - ---- x
                                                   11


                                           3          5          7          9
                 x          -120 x + 2240 x  - 10752 x  + 18432 x  - 10240 x
      D - --------------- - -------------------------------------------------]
          (x - 1) (x + 1)           2        4         6         8         10
                            1 - 60 x  + 560 x  - 1792 x  + 2304 x  - 1024 x

\end{verbatim}
%}
\noindent
Here, the two components of \textbf{v} are the two possible
right hand factors of $L$.  By rerunning the commands with
other (positive integer) values of $n$, the growth in the
factors can be observed.  Other classical operators from the
1800's, such as those of Hermite, Legendre, and Laguerre, have
similar behaviour.

\end{document}